\documentclass[10pt,a4paper]{amsart}

\usepackage{latexsym,amssymb,amsmath,amsthm,amsfonts,enumerate,verbatim,xspace,
exscale}
\usepackage{graphicx}
\usepackage{color,amsbsy}
 
\usepackage{hyperref}
\definecolor{darkblue}{rgb}{0,0,.5}
\hypersetup{colorlinks=true, breaklinks=true, linkcolor=darkblue, menucolor=darkblue,  urlcolor=darkblue,citecolor=darkblue}

\usepackage{setspace}

\theoremstyle{plain}
\newtheorem{theorem}{Theorem}[section]

\newtheorem{proposition}[theorem]{Proposition}

\theoremstyle{definition}

\newtheorem{remark}[theorem]{Remark}

\def\D{\mathcal{D}}

\def\R{\mathbb{R}}
\def\I{\mathbb{I}}

\def\F{\mathcal{F}}

\def\A{\mathcal{A}}
\def\P{\mathcal{P}}
\def\G{\mathcal{G}}

\def\JG{J_{\mathcal{G}}}
\def\JD{J_{\mathcal{D}}}
\def\JP{J_{\mathcal{P}}}

\newcommand{\up}{\upshape}

\def\vv<#1>{\langle#1\rangle}
\def\ww<#1>{\langle\langle#1\rangle\rangle}

\newcommand{\diag}[1]{\mbox{$\textup{diag}(#1)$}}

\providecommand{\del}{\partial}

\newcommand{\om}{\omega}
\newcommand{\Om}{\Omega}

\newcommand{\lam}{\lambda}

\newcommand{\by}[2]{\mbox{$\frac{#1}{#2}$}}
\newcommand{\cinf}{\mbox{$C^{\infty}$}}

\newcommand{\du}{\mathfrak{d}}

\newcommand{\gu}{\mathfrak{g}}

\newcommand{\po}{\mathfrak{p}}

\newcommand{\Ad}{\mbox{$\text{\upshape{Ad}}$}}
\newcommand{\ad}{\mbox{$\text{\upshape{ad}}$}}

\newcommand{\SO}{\mbox{$\textup{SO}$}}
\newcommand{\so}{\mbox{$\mathfrak{so}$}}

\newcommand{\X}{\mbox{$\mathcal{X}$}}

\newcommand{\revise}[1]{#1}

\title[ ]{%
Nonlinear feedback, double bracket dissipation and port control of Lie-Poisson systems}

\author{Simon Hochgerner}
\address{Austrian Financial Market Authority (FMA),
Otto-Wagner Platz 5, A-1090 Vienna
}
\email{simon.hochgerner@fma.gv.at} 

\begin{document}

\begin{abstract}
Methods from controlled Lagrangians, double bracket dissipation and  interconnection and damping assignment -- passivity based control (IDA-PBC) are used to construct nonlinear feedback controls which (asymptotically) stabilize  previously unstable equilibria of Lie-Poisson Hamiltonian systems. The results are applied to find an asymptotically stabilizing control for the rotor driven satellite, and a stabilizing control for Hall magnetohydrodynamic flow.
\end{abstract}

\maketitle


\section{Introduction}\label{sec0:intro}
This paper is concerned with two strands of feedback control: the theory of controlled Lagrangians initiated in \cite{K85} and further elaborated by \cite{BLM01a,BLM01b}, and the method of interconnection and damping assignment -- passivity based control (IDA-PBC) developed by \cite{OSME02,OG04,BCO16} in the context of port controlled Hamiltonian systems. The relation between these approaches has been studied in \cite{BOS02,CBLMW02}, and it has been shown, in the framework of almost Poisson structures, that these techniques are equivalent when regarded in sufficient generality. 

The controlled Lagrangian/Hamiltonian approach was first formulated in the framework of Hamiltonian mechanics (e.g., \cite{BMS97}) and later studied systematically from the Lagrangian point of view (\cite{BLM01a,BLM01b}). Both perspectives are equivalent (for the kinetic energy functions considered in this paper) and will be referred to collectively as the controlled Lagrangian (CL) approach. 

The CL method is particularly well suited to symmetry arguments and the tools of geometric mechanics (\cite{WK92,BKMS92,BMS97,BCLMW00,CM04}). 
The original idea of CL is to start from a mechanical system defined by a kinetic energy function of Kaluza-Klein type and to modify this structure (kinetic energy shaping) so that the corresponding equations of motion have force terms only in a given symmetry direction. In \cite{H19,H20} this procedure has been reversed: associated to a given Lie group action is a force acting in symmetry directions  and the goal is to reformulate this as a mechanical system associated to a new kinetic energy function. When such a function can be found, these approaches are clearly equivalent. However, as shown presently, more general forces along symmetry directions may be considered and this leads to not only a new energy function but also a new almost Poisson structure and a new symmetric tensor responsible for energy dissipation. This fits the framework of IDA-PBC, and has the additional feature of being adapted to the symmetry methods of geometric mechanics. 

This paper deals with Kaluza-Klein Hamiltonian systems in the Lie-Poisson setting. This implies that the phase space of the mechanical system is  a direct product dual Lie algebra $\po^*=\du^*\times\gu^*$ where $\du$ and $\gu$ are Lie algebras. Both, $\du$ and $\gu$, may be finite or infinite dimensional and non-Abelian. The component $\gu^*$ represents the symmetry directions where the actuating force is applied.
The contributions, which are new to the best of the author's knowledge, are the following: 
\begin{enumerate}
    \item 
    In Proposition~\ref{prop:matching}   matching conditions are derived to characterize the relationship between symmetry actuation and modified Kaluza-Klein Hamiltonians. These generalize the Euler-Poincar\'e matching conditions of \cite{BLM01b} to the non-Abelian and coordinate independent setting (Section~\ref{sec:EPcond}).  
    \item
    Section~\ref{sec:ConDis} introduces a nonlinear feedback control based on double bracket dissipation (\cite{BKMR96}). This yields controlled energy dissipation. In Theorem~\ref{thm:FB_IDA} this construction is recognized as an IDA-PBC system and the corresponding structure consisting of a new Hamiltonian, a new almost Poisson tensor and a new symmetric tensor is given explicitly. This method is fully constructive, and depends neither on the matching conditions of CL theory (\cite{BLM01a,BLM01b}) nor on the solution of a PDE (\cite{OSME02}). In fact, the Hamiltonian found in this manner is of fourth order and does not satisfy the matching conditions (i.e.\ Proposition~\ref{prop:matching}). 
    \item 
    Theorem~\ref{thm:FB_IDA} is applied in Section~\ref{sec2:sat} to asymptotically stabilize the satellite with a rotor in the full phase space. This conclusion is the same as in \cite{BCLMW00} but reached through a different method. 
    \item
    Section~\ref{Sec3} contains an application of Theorem~\ref{thm:FB_IDA} to Hall magnetohydrodynamic (MHD) flow. This is used in Section~\ref{Sec4:shear} to stabilize a previously unstable equilibrium of planar shear flow. As in Section~\ref{sec2:sat}, stability is achieved in the full phase space of fluid velocity and magnetic field variables, but asymptotic stability cannot be concluded in this case. 
\end{enumerate}

The main idea underlying these points is that the under-actuation associated to the the given symmetry group direction and the double bracket construction of \cite{BKMR96} give rise to a new conserved quantity when the dynamics are considered as a mechanical system on the cotangent bundle $T^*\mathcal{P}$, where $\mathcal{P}$ is the Lie group of $\po$. This in turn leads to a nonlinear coordinate transformation $\Phi: \po^*\to\po^*$, $(\nu,\alpha)\mapsto(\nu,\beta)$. In the new coordinates, there is a natural candidate, $g_C$, for a modified kinetic energy function. Therefore, the IDA-PBC problem can be solved algebraically, 
and we obtain a controlled interconnection and damping structure $\Pi_C$ and $\mathcal{R}_C$. This structure can now be transformed to the original coordinates via pullback to yield the desired structure: the fourth order Hamiltonian function $h_d=\Phi^*g_C$, the almost Poisson tensor $\Pi_d=\Phi^*\Pi_C$ and the symmetric tensor $\mathcal{R}_d = \Phi^*\mathcal{R}_C$, cf.\ Theorem~\ref{thm:FB_IDA}. This construction relies on the controlled conserved quantity \revise{$J_{\mathcal{G}}+j$} described in Proposition~\ref{prop:Cons}, it is therefore a synthesis of CL and the approach via canonical transformations to algebraic solutions of IDA-PBC outlined in \cite{FS01}. 

Moreover, the closed-loop dynamics, associated to the feedback control, appear in the $(\nu,\beta)$-variables in a particularly appealing format: the double bracket dissipation acts on the $\nu$-variables such that coadjoint orbits in $\du^*$ are preserved; on the other hand, in $\beta$ it is desirable to asymptotically restore the origin, corresponding to the closed loop $\alpha =  -C\nu - N(\nu)$, and accordingly decay of coadjoint orbits in $\gu^*$ is incorporated; the relevant control force is specified in terms of a linear map $C: \du^*\to\gu^*$ leading to a new inner product on $\du^*$, and a nonlinear map $N: \du^*\to\gu^*$ encoding the double bracket structure.

Other recent approaches to algebraic IDA-PBC include \cite{BCO16} and \cite{MSKS17}. The latter also introduces the notion of geometric PBC, but assumes a finite dimensional set-up and that the symmetry group action is Abelian. The method proposed in this paper differs from these approaches in the following ways: the coadjoint orbit structure in $\du^*$ is used to obtain double bracket dissipation; the controlled momentum map yields a nonlinear diffeomorphism such that the transformed closed-loop system can be  stabilized in the full phase space (non-actuated and actuated variables) by means of the energy-Casimir method; non-Abelian symmetry groups are allowed; and the construction is coordinate independent so that it may be applied also to infinite dimensional problems such as Hall MHD flow.  

This paper is restricted to mechanical systems of Lie-Poisson type which includes a large class of systems (\cite{BLM01a,BLM01b,HSS09,MRW84}). However, because of independence from (Euler-Poincar\'e or Lie-Poisson) matching conditions and since the controlled momentum map exists for symmetry actuation on any cotangent bundle (Proposition~\ref{prop:Cons}), the construction in Theorem~\ref{thm:FB_IDA} may carry over to more general phase spaces $T^*\mathcal{P}$ where $\mathcal{P}\to\mathcal{D}$ is a principal fiber bundle with structure group $\mathcal{G}$.

\section{Feedback control of Lie-Poisson systems}\label{sec1:feedback}

\subsection{Lie-Poisson systems}
Let $\D$ and $\G$ be finite or infinite dimensional Lie groups and consider the direct product Lie group $\P := \D\times\G$. In the following $\D$ will be thought of as the base and $\G$ as the symmetry group so that the projection $\P\to\D$ is a principal fiber bundle with structure group $\G$. 

For example, in Section~\ref{sec2:sat}, $\D$ will correspond to $\SO(3)$ describing the satellite's orientation and $\G$ will coincide with $S^1$ to represent the rotor's angle\revise{, cf.\ also Remark~\ref{rem:sat1} below.} 

Tangent and cotangent bundles can be represented via left or right multiplication as $T\P\cong \P\times\po = \D\times\G\times\du\times\gu$ and $T^*\P \cong \P\times\po^* = \D\times\G\times\du^*\times\gu^*$ where $\po$, $\du$, $\gu$ denote the Lie algebras associated to $\P$, $\D$, $\G$ respectively. The corresponding dual Lie algebras, $\po^*$, $\du^*$, $\gu^*$, are defined as smooth duals (also called: the regular part) such that there exist  isomorphisms   $\po\cong\po^*$, $\du\cong\du^*$, $\gu\cong\gu^*$, even in the infinite dimensional case (\cite{Michor06}). The left trivialization is called body representation (often used for examples involving rigid bodies), and the right trivialization is called space representation (often used for fluid dynamical examples). See \cite{HSS09}. 

From now on it is assumed that one such representation, left or right, has been chosen. Whenever the representation induces a specific sign this will be marked by $\pm$ where $+$ corresponds to `left' and $-$ to `right', and vice versa for $\mp$. Thus we can identify $T^*\P=\P\times\po^*$ and elements $\Pi\in T^*\P$ are written as $\Pi=(\phi,g,\nu,\alpha)$. 

Let $M^{\pm}$ denote left (or right) multiplication in $\G$ and $\D$, i.e.\ $M^+_{\psi}\phi = \psi\phi$ and $M^-_{\psi}\phi = \phi\psi$ in $\G$. The variational derivative, $\delta f/\delta\nu\in\du$, of a function $f: \du^*\to\R$ is defined by the dual pairing
\[
 \left<df(\nu), \eta \right>_{\mathfrak{d}^{**}\times\mathfrak{d}^*} 
 = \frac{\del}{\del t}\Big|_0 f(\nu+t\eta)
 = \left<\eta, \frac{\delta f}{\delta\nu}\right>_{\mathfrak{d}^{*}\times\mathfrak{d}} 
\]
where $\nu,\eta\in\du^*$, $df$ is the exterior derivative (full differential) and the source space of the pairing is indicated. The variational derivative of the pull back $\Phi^*f = f\circ\Phi$ along a function $\Phi:\du^*\to\du^*$ is given by 
\begin{equation}
    \label{eS1:pullb}
    \frac{\delta\Phi^*f}{\delta\nu}
    = (T\Phi)^*\frac{\delta f}{\delta\Phi(\nu)}. 
\end{equation}

The conjugation by left and right multiplication in $\G$ gives rise to the adjoint action $\Ad(g):\gu\to\gu$. Indeed, let $\textup{Conj}_g h = ghg^{-1}$ for $g,h\in\G$. Then $\Ad(g) := T_e\textup{Conj}_g: \gu\to\gu$, and the coadjoint action $\Ad(g)^*: \gu^*\to\gu^*$ is given by $\Ad(g)^*\alpha = \alpha\circ\Ad(g)$. Moreover, we define $\ad(X)Y = \del \Ad(\textup{exp}(tX))Y/\del t|_0$, and $\ad(X)^*\alpha = \alpha\circ\ad(X)$. The corresponding operations associated to $\P$ and $\D$ will be denoted by the same symbols. In the infinite dimensional case these operations have to be restricted to their appropriate domains of existence (\cite{Michor06}). 

On the cotangent bundle  $T^*\P$ there exists a canonical exact symplectic form, $\Om^{T^*\mathcal{P}}$. Let $h: T^*\P\to\R$ be a Hamiltonian function which is left (or right) invariant with respect to multiplication, $M^{\pm}$, of $\P$ on itself. The corresponding equations of motion, in left (or right) representation, are given by 
\begin{align}
\label{e:eom}
 \Big(
  TM^{\pm}_{\phi^{-1}}\dot{\phi} , TM^{\pm}_{g^{-1}}\dot{g} , \dot{\nu},\dot{\alpha}
 \Big)
    = 
    X_h(\Pi)
    =
    \Big(
    \by{\delta h}{\delta\nu},\by{\delta h}{\delta\alpha},
    \pm\ad(\by{\delta h}{\delta\nu})^*\nu, \pm\ad(\by{\delta h}{\delta\alpha})^*\alpha
    \Big)
\end{align}
where $X_h = (\Om^{T^*\mathcal{P}})^{-1}dh$ is the Hamiltonian vector field,
$\Pi = (\phi,g,\nu,\alpha)\in \D\times\G\times\du^*\times\gu^*$. 
The canonical momentum map $\JP: T^*P\to\po^*$ corresponding to left (or right) multiplication is:
\[
 \JP = (\JD , \JG):
 T^*\P\to\du^*\times\gu^*,\qquad
 (\phi,g,\nu,\alpha)\mapsto
 (\Ad(\phi^{\mp})^*\nu, \Ad(g^{\mp})^*\alpha )
\]
See \cite{MRW84}. 

For the remainder of this section, it is assumed that $h:T^*\P\to\R$ is left (or right) invariant.

\begin{remark}[Satellite with a rotor]\label{rem:sat1}
The prime example for the theory developed in this section is that of the rigid body with a rotor attached to one of the principal axis. This example is treated in detail in Section~\ref{sec2:sat} where bibliographical references are provided. However, to make the abstract formulation of the Lie-Poisson setting more accessible, the main points will be sketched for the satellite example also in this section. For this example, $\D=\SO(3)$ corresponds to the rigid body and $\G=S^1$ corresponds to the rotor attached to the short axis. We use the body representation whence the top sign is the relevant one. The Lie algebra structure on $\du=\so(3)$, which is identified with $\R^3$, is given by the cross product, i.e.\ $\ad(u)v = u\times v$, and the bracket on $\gu=\R$ is the Abelian, i.e.\ trivial, one. Thus the equations of motion~\eqref{e:eom} for  a curve $(\nu_t,\alpha_t)\in\R^3\times\R$ are given by  $\dot{\nu} = -\frac{\delta h}{\delta \nu}\times \nu$ and $\dot{\alpha}=0$. 
\end{remark}

\subsection{Controlled equations of motion} 
Let $G: \gu^*\to\gu^*$ be an isomorphism, which is assumed to be linear but not necessarily equivariant. Let $C: \du^*\to\gu^*$ be a linear map. Let $N: \du^*\to\gu^*$ be a smooth map, which may be nonlinear. 
Following the construction in \cite{H20}, we define a map $j: T^*\P\to\gu^*$ through
\begin{equation}
    \label{e:def_j}
    j(\phi,g,\nu,\alpha)
    = 
    \Ad(g^{\mp})^* G^{-1} \Big( C\nu + N(\nu) + (1-G)\alpha \Big).
\end{equation}
Associated to $j$ is a control force $f: T^*\P\to\gu^*$ via
\begin{equation}
    \label{e:def_f}
    f(\Pi)
    = -\Ad(g^{\mp})^*\,G\, \Ad(g^{\pm})^* \, dj \, X_h(\Pi)
\end{equation}
and a corresponding closed-loop control (in the terminology of \cite[Def.~3.1]{CBLMW02}) $\mathcal{U}: T^*\P\to T^*\P$ as
\begin{equation}
    \label{e:def_U}
    \mathcal{U}(\Pi)
    = \Big(\phi,g,0,\Ad(g^{\pm})^*f(\Pi)\Big). 
\end{equation}
The control force $F\in\X(T^*\P)$, which is a vector field on $T^*\P$, is now defined as the vertical lift of $\mathcal{U}$ at $\Pi$, that is
\begin{equation}
    \label{e:def_F}
    F(\Pi)
    := 
    \frac{\del}{\del s}\Big|_{s=0}\Big(\Pi + s\, \mathcal{U}(\Pi)\Big)
    = 
    \Big(0,0,0,\Ad(g^{\pm})^* f(\Pi)\Big)
    \in T_{\Pi}T^*\P 
\end{equation}
and the controlled equations of motion are
\begin{equation}
    \label{e:ceom1}
    \dot{\Pi}
    = X_h(\Pi) + F(\Pi).  
\end{equation}

\begin{proposition}\label{prop:Cons}
The following are true.
\begin{enumerate}[\up (1)]
    \item 
    $F$ is vertical with respect to the cotangent projection $\tau: T^*\P\to \P$, i.e.\ $T\tau\, F = 0$. 
    \item 
    $d\JD\,F=0$ and $d\JG\,F = f$.
    \item 
    The momentum map $\JD$ and the controlled momentum map $\JG+j$ are conserved along solutions of equation~\eqref{e:ceom1}: 
    \begin{align}
        d\JD\,(X_h + F) &= 0 \label{e:consJM} \\
        d(\JG+j)\,(X_h + F) &= 0 \label{e:consJGj}
    \end{align}
\end{enumerate}
\end{proposition}

\begin{proof}
Items (1) and (2) follow because $F$ is defined as a vertical lift of $\mathcal{U}$. Equation~\eqref{e:consJM} is $d\JD\, X_h = 0$ which holds since $h$ is, in particular, $\D$-invariant. Equation~\eqref{e:consJGj} is shown in \cite{H20} under the assumption of $j: T^*\P\to\gu^*$ being fiber linear. However, the calculation extends to the present case (because linearity of $G$ suffices): 
\begin{align*}
    dj\, F(\Pi)
    &= 
    \frac{\del}{\del s}\Big|_{s=0}\,j\Big(\Pi + s\,\mathcal{U}(\Pi)\Big)
    = 
    \frac{\del}{\del s}\Big|_{s=0}\,
    \Ad(g^{\mp})^*G^{-1}
    \Big( C\nu+N(\nu) + (1-G)(\alpha+s\,\Ad(g^{\pm})^*\,f) \Big)
    \\
    &= 
    -dj\,X_h(\Pi) + \Ad(g^{\mp})^*G\Ad(g^{\pm})^* dj\, X_h(\Pi) 
    = 
    -dj\,X_h(\Pi) - f(\Pi) 
\end{align*}
whence
$
    d(\JG+j)\,(X_h + F)
    = 
    0 + f + dj\, X_h + dj\,F = 0
$.
\end{proof}

To express \eqref{e:ceom1} in the Lie-Poisson setting, note that \eqref{e:eom} yields 
\begin{align*}
    dj\, X_h(\Pi)
    &= 
    \mp\Ad(g^{\mp})^*\Big(
     \ad(X)^*G^{-1}(C\nu + N(\nu) + \alpha) 
     - G^{-1}C\,\ad(u)^*\nu-G^{-1}\,dN(\nu)\,\ad(u)^*\nu 
      - G^{-1}\ad(X)^*\alpha 
     \Big) 
\end{align*}
where $u=\delta h/\delta\nu$ and $X=\delta h/\delta\alpha$. 
The corresponding force in the Lie-Poisson formulation  is
\begin{align}
\label{eS1:tildeU} 
    \mathcal{U}_{\textup{LP}} (\Pi)
    &:= 
    \Ad(g^{\pm})^*f(\Pi) 
    = -G\Ad(g^{\pm})^*\,dj\,X_h(\Pi) \\ 
    &=
    \notag 
    \pm\Big(
     G\,\ad(X)^*G^{-1}(C\nu + N(\nu) + \alpha)  
     - C\,\ad(u)^*\nu - dN(\nu)\,\ad(u)^*\nu 
     - \ad(X)^*\alpha 
    \Big) 
\end{align}
which is independent of the base point, $\mathcal{U}_{\textup{LP}} (\Pi) = \mathcal{U}_{\textup{LP}} (\nu,\alpha)$, as it should be in order for \eqref{e:ceom1} to be left (or right) invariant.  Hence \eqref{e:ceom1} may be expressed equivalently in the left (or right) representation as 
\begin{align}
    \label{eS1:ceom1LP} 
    \dot{\nu} 
    &= \pm\ad(u)^*\nu  \\ 
    \dot{\alpha} 
    &= \pm\ad(X)^*\alpha + \mathcal{U}_{\textup{LP}} (\nu,\alpha)
    = \pm\Big( G\,\ad(X)^*G^{-1}(C\nu + N(\nu) + \alpha)  
     - C\,\ad(u)^*\nu - dN(\nu)\,\ad(u)^*\nu \Big) 
 \label{eS1:ceom2LP}
\end{align}
where $u = \delta h/\delta\nu = TM_{\phi^{-1}}\dot{\phi}$ and $X=\delta h/\delta\alpha = TM_{g^{-1}}\dot{g}$,  and $M$ is left (or right) multiplication in the group.

\begin{remark}[Satellite with a rotor]\label{rem:sat2}
Continuing the example from Remark~\ref{rem:sat1}, equation~\eqref{eS1:ceom1LP} remains unchanged as $\dot{\nu} = \delta h/\delta\nu \times\nu$. The equation for $\alpha$ now becomes the controlled equation of motion $\dot{\alpha} = C\,u\times\nu + dN(\nu)\,u\times\nu$ where $C$ is a $1\times 3$ matrix and $N:\R^3\to\R$ is a smooth map.
\end{remark}

\subsection{Closed loop equations}
let $\beta := G^{-1}(C\nu + N(\nu) + \alpha)$. 
Equation~\eqref{e:consJGj} implies that the relation 
\begin{equation}
    \label{e:feedback}
    (\JG+j)(\Pi)
    = \Ad(g^{\mp})^*G^{-1}\Big(C\nu + N(\nu) + \alpha\Big) 
    = \Ad(g^{\mp})^*\beta 
    = \beta_0
\end{equation}
is conserved along solutions of \eqref{e:ceom1} where $\beta_0\in\gu^*$ is constant. 
This gives an expression for $\alpha = G\beta - C\nu - N(\nu)$ as a function of $\nu$ and $\beta$. Feeding this expression back into $\delta h/\delta \nu$ and $\delta h/\delta\alpha$, for all occurrences of $\alpha$, yields 
\begin{align}
    \label{e:ceom2}
    \dot{\nu} &= \pm\ad(u)^*\nu, \\
    \label{e:ceom3}
    \dot{\beta} &= \pm\ad(X)^*\beta 
\end{align}
where $u = \delta h/\delta\nu = TM_{\phi^{-1}}\dot{\phi}$ and $X=\delta h/\delta\alpha = TM_{g^{-1}}\dot{g}$ are now expressed in terms of $\nu$ and $\beta$. Equations~\eqref{e:ceom2}-\eqref{e:ceom3} are referred to as the closed-loop equations associated to the feedback control \eqref{e:ceom1} or, equivalently in Lie-Poisson form, \eqref{eS1:ceom1LP}-\eqref{eS1:ceom2LP}.   

\begin{remark}
If $u$ and $X$ can be written as functional derivatives with respect to a function depending on $\nu$ and $\beta$ then \eqref{e:ceom2}-\eqref{e:ceom3} are in Hamiltonian form, and in fact even Lie-Poisson. But this is not always possible. Proposition~\ref{prop:matching} below gives conditions under which such a Hamiltonian, $h_C$, can be found for the case $N=0$. Specific choices of $N\neq0$ can then be used to add dissipation to the the Hamiltonian system defined by $h_C$. 
\end{remark} 

Following \cite{BLM01a,BLM01b,BKMS92,BMS97}, we assume from now on that the Hamiltonian $h_0$ is not only left (or right) $\P$-invariant but also of Kaluza-Klein type. This means that $h_0$ is the kinetic energy of a $\P$-left- (or right-) invariant Kaluza-Klein metric $\mu_0^{\mathcal{P}}$, i.e.\ $h_0(\Pi) = \vv<(\nu,\alpha)^{\top}, (\mu_0^{\mathcal{P}})^{-1}(\nu,\alpha)^{\top} >/2$ where $\vv<.,.>$ is the duality pairing.   
Thus $\mu_0^{\mathcal{P}}$ is assumed to be a Kaluza-Klein metric on the principal fiber bundle $\P\to\D$ where the symmetry group $\G$ acts by left (or right) multiplication on itself. The metric $\mu_0^{\mathcal{P}}$ is determined by the following data: 
\begin{enumerate}[\up (1)]
    \item
    A Riemannian metric on $\D$ which is invariant with respect to the left (or right) multiplication of $\D$ on itself. Hence the metric is determined by an isomorphism $\mu_0: \du\to\du^*$ such that the bilinear form $\vv<\mu_0 .,.>$ is symmetric and positive definite.  
    \item 
    An isomorphism $\I_0: \gu\to\gu^*$ such that the associated bilinear form on $\gu$ is symmetric and positive definite.
    \item 
    A principal bundle connection $T\P\to\gu$, $(\phi,g,u,X)\mapsto \Ad(g^{\pm})(X + \A_0u)$ determined by a linear map $\A_0: \du\to\gu$.
\end{enumerate}
Identifying the metric with the induced isomorphism to the (regular) dual, 
this means that $\mu_0^{\mathcal{P}}: \po\to\po^*$ is given in matrix form: 
\begin{equation}
\label{eS1:KK} 
 \mu_0^{\mathcal{P}}
 =
 \left(
 \begin{matrix}
  \mu_0+\A_0^*\I_0\A_0 & \A_0^*\I_0 \\ 
  \I_0\A_0  & \I_0 
 \end{matrix}
 \right) 
\end{equation}
This entails
\begin{equation}
    \label{e:uX}
    \Big(\frac{\delta h_0}{\delta\nu}, \frac{\delta h_0}{\delta\alpha}\Big)
    = \Big(u, X\Big)
    = (\mu_0^{\mathcal{P}})^{-1} \Big(\nu,\alpha\Big) 
    = \Big( \mu_0^{-1}(\nu-\A_0^*\alpha) , \I_0^{-1}\alpha-\A_0 u \Big) .
\end{equation}

\begin{proposition}\label{prop:matching}
Assume that $N=0$ and let $h_0$ be the Hamiltonian associated to the Kaluza-Klein metric \eqref{eS1:KK}. Then, by virtue of $\beta = G^{-1}(C\nu+\alpha)$, equation~\eqref{e:ceom1} is equivalent to \eqref{e:ceom2}-\eqref{e:ceom3} where
\begin{align*}
    u 
    &= \mu_0^{-1}(1+\A_0^*C)\nu - \mu_0^{-1}\A_0^*G\beta \\ 
    X
    &= \I_0^{-1}G\beta + \A_0\mu_0^{-1}\A_0^*G\beta 
      - \I_0C\nu - \A_0\mu_0^{-1}(1+\A_0^*C)\nu
\end{align*}
Moreover, the following are equivalent:
\begin{enumerate}
    \item 
    Equations~\eqref{e:ceom2}-\eqref{e:ceom3} are Lie-Poisson on the direct product $\po^*=\du^*\times\gu^*$ with respect to a kinetic energy Hamiltonian associated to a Kaluza-Klein inner product $\mu_C^{\mathcal{P}}$ on $\po$.
    \item
    The following conditions are satisfied: 
    \begin{itemize}
    \item
    $1+\A_0^*C: \du^*\to\du^*$ is invertible and $\mu_C := (1+\A_0^*C)^{-1}\mu_0$ defines an inner product on $\du$. 
        \item 
        $G$ satisfies $\A_0^*G = \mu_0\mu_C^{-1}\A_C^*$ where $\A_C := \A_0+\I_0^{-1} C \mu_C$.
        \item 
        $G-C\A_C^*$ is invertible and  $\I_C := (G-C\A_C^*)^{-1}\I_0$ defines an inner product on $\gu$.
    \end{itemize}
    In this case,  $\mu_C^{\mathcal{P}}$ corresponds to $\mu_C$, $\I_C$ and $\A_C$ as in \eqref{eS1:KK}, and 
    the Hamiltonian is $h_C: \du^*\times\gu^*\to\R$, $(\nu,\beta)\mapsto\vv<(\nu,\beta)^{\top},(\mu_C^{\mathcal{P}})^{-1}(\nu,\beta)^{\top}>/2$. 
    \item 
There exist inner products $\mu_C$, $\I_C$ on $\du$, $\gu$, respectively, and a linear map $\A_C: \du\to\gu$ such that the following conditions are satisfied: 
\begin{align}
    \tag{LP1}\label{LP1}
    \I_C\A_C 
    &= \I_0\A_0 
    \\ 
    \tag{LP2}\label{LP2}
    \mu_C+\A_C^*\I_C\A_C
    &= \mu_0+\A_0^*\I_0\A_0
\end{align}
In this case,  $\mu_C^{\mathcal{P}}$ corresponds to $\mu_C$, $\I_C$ and $\A_C$, as in \eqref{eS1:KK}, and 
the Hamiltonian is $h_C: \du^*\times\gu^*\to\R$, $(\nu,\beta)\mapsto\vv<(\nu,\beta)^{\top},(\mu_C^{\mathcal{P}})^{-1}(\nu,\beta)^{\top}>/2$. Moreover, $C$ and $G$ are recovered by $C = \I_0 (\A_C-\A_0) \mu_C^{-1}$ and $G = \I_0\,\I_C^{-1} + C\,\A_C^*$. 
\end{enumerate}
\end{proposition}

\begin{proof}
It remains to show equivalence of items (1), (2) and (3). 


Equivalence of (1) and (2) follows from the observation that the closed-loop system \eqref{e:ceom2}-\eqref{e:ceom3} can be written as a Lie-Poisson system with respect to a kinetic energy Hamiltonian (of Kaluza-Klein type) if, and only if, 
\[ 
 \left(
  \begin{matrix}
   u\\
   X 
  \end{matrix}
 \right)
 = 
 \Big(\mu_0^{\mathcal{P}}\Big)^{-1} 
  \left(
  \begin{matrix}
   \nu \\
   \alpha  
  \end{matrix}
 \right)
 = 
  \Big(\mu_0^{\mathcal{P}}\Big)^{-1} 
  \left(
  \begin{matrix}
   1 & 0 \\
   -C & G  
  \end{matrix}
 \right)
  \left(
  \begin{matrix}
   \nu \\
   \beta   
  \end{matrix}
 \right)
 =   
  \Big(\mu_C^{\mathcal{P}}\Big)^{-1} 
  \left(
  \begin{matrix}
   \nu \\
   \beta   
  \end{matrix}
 \right)
 = 
    \left(
    \begin{matrix}
     \mu_C^{-1}  &  -\mu_C^{-1}\A_C^* \\
     -\A_C\mu_C^{-1} & \I_C^{-1} + \A_C\mu_C^{-1}\A_C^* 
    \end{matrix}
    \right)
      \left(
  \begin{matrix}
   \nu \\
   \beta   
  \end{matrix}
 \right)
\] 
This implies the following equations: $\mu_C^{-1} = \mu_0^{-1}(1+\A_0^*C)$, $\mu_C^{-1}\A_C^* = \mu_0^{-1}\A_0^* G$, $\A_C = \A_0+\I_0^{-1} C \mu_C$ and $\I_C^{-1} = \I_0^{-1}(G - C\A_C^*)$. 

Equivalence of (2) and (3).  
The above matrix equation implies $\nu = (\mu_C+\A_C^*\I_C\A_C)u + \I_C\A_C X = (\mu_0+\A_0^*\I_0\A_0)u + \I_0\A_0 X$ which splits into two equations, namely \eqref{LP1} and \eqref{LP2}, for $X$ and $u$. Thus  \eqref{LP1} and \eqref{LP2} follow if \eqref{e:ceom2}-\eqref{e:ceom3} are Lie-Poisson with respect to $\mu_C^{\mathcal{P}}$. 

Conversely, if \eqref{LP1} and $\eqref{LP2}$ are satisfied such that $C = \I_0 (\A_C-\A_0) \mu_C^{-1}$ and $G = \I_0\,\I_C^{-1} + C\,\A_C^*$. Then it follows that $\mu_C = (1+\A_0^* \, C)^{-1}\,\mu_0$, $\mu_0^{-1}\,\A_0^*\,G = \mu_C^{-1}\A_C^*$, $\A_0+\I_0^{-1}\,C\,\mu_C = \A_C$ and $(G-C\,\A_C^*)^{-1}\,\I_0 = \I_C$.
\end{proof}

\begin{remark}\label{rem:constr}
Items (2) and (3) are both constructive in the sense that they algebraically specify the conditions which the control force, $F$, or the conditions which the controlled metric, $\mu_C^{\mathcal{P}}$, has to satisfy in order to obtain an equivalent Lie-Poisson system. The second point of view is the one that is also taken in the theory of controlled Hamiltonians (or Lagrangians) in \cite{BLM01a,BLM01b}. 
\end{remark}

\begin{remark}\label{rem:beta0}
Let $\beta=0$. Then equation~\eqref{e:ceom2} is Lie-Poisson on $\du^*$ with respect to a kinetic energy Hamiltonian if, and only if, $1+\A_0^*C: \du^*\to\du^*$ is invertible and $\mu_C := (1+\A_0^*C)^{-1}\mu_0$  defines an inner product on $\du$.  In this case, the Hamiltonian is defined by  $h_C: \du^*\to\R$, $\nu\mapsto\vv<\nu,\mu_C^{-1}\nu>/2$ and $u=\delta h_C/\delta\nu$. Hence, for $\beta=0$, the conditions on $\I_C$ and $\A_C$ are void, and one may use $G=1$ in definition \eqref{e:def_j}. 
\end{remark}

\begin{remark}[Satellite with a rotor]\label{rem:sat3}
Continuing the discussion from Remark~\ref{rem:sat2}, we specify $\mu_0 = \diag{\lam_1,\lam_2,I_3}$ with $I_3>\lam_2>\lam_1>0$ and $\I_0 = i_3>0$. Further, the connection form is given by the projection $\A_0 = e_3^{\bot}: \R^3\to\R$, $(u^1,u^2,u^3)\mapsto u^3$. 
The Hamiltonian $h_0: \du^*\times\gu^* = \R^3\times\R\to\R$ is consequently given in Kaluza-Klein form as
\[
 h_0(\nu,\alpha)
 = \frac{1}{2}\vv<(\mu_0^{\mathcal{P}})^{-1}(\nu,\alpha),(\nu,\alpha)>
 = 
 \frac{1}{2}
 \left\langle 
 \left(
 \begin{matrix}
  \diag{\lam_1^{-1}, \lam_2^{-1}, I_3^{-1}} & -I_3^{-1} e_3 \\ 
  -I_3^{-1}e_3^{\bot} & i_3^{-1} + I_3^{-1}
 \end{matrix}
 \right)
  \left(
 \begin{matrix}
  \nu \\ 
  \alpha 
 \end{matrix}
 \right)  , 
  \left(
 \begin{matrix}
  \nu \\ 
  \alpha 
 \end{matrix}
 \right)
 \right\rangle 
\]
Thus we may use the formulas in item (2) of Proposition~\ref{prop:matching} to construct feedback controlled equations of motion which are in Lie-Poisson form. To do so it is necessary that $C: \du^*=\R^3\to\R$ is such that  $1+\A_0^*C = \diag{1,1,1} + e_3C$ is a symmetric and positive definite matrix. This is the case if 
$C 
 = k \I_0\A_0\mu_0^{-1} 
 = k \frac{i_3}{I_3} e_3^{\bot}$ 
for some $k>-I_3/i_3$, since then $1+\A_0^*C = \diag{1,1,1 + k \frac{i_3}{I_3}}$.  
This allows to obtain the controlled metric tensor as $\mu_C = \diag{\lam_1, \lam_2, (1 + k i_3/I_3)^{-1} I_3}$. 
The other data can now be calculated as
\begin{align*}
    \A_C
    &= 
    \A_0 + \I_0^{-1}C\Big(1+\A_0^*C\Big)^{-1}\mu_0
    = \Big(1+\frac{k }{1+k}\Big)e_3^{\top} \\
    G 
    &= 
    1+k\frac{i_3 + I_3}{I_3} \\ 
    \I_C 
    &= 
    \Big(G-C\A_C^*\Big)^{-1}\I_0
    = \frac{i_3(I_3+ki_3)}{I_3+ki_3 + kI_3}
\end{align*}
Therefore, Proposition~\ref{prop:matching} implies that the Kaluza-Klein metric $\mu_C^{\mathcal{P}}$ associated to $\mu_C$, $\I_C$ and $\A_C$ satisfies the matching conditions. That is, the closed loop equations \eqref{e:ceom2}-\eqref{e:ceom3}, with $N=0$ and  $\beta := G^{-1}(C\nu+\alpha)$, are  given by 
\begin{align*}
    \dot{\nu} 
    &= \ad(u)^*\nu 
    = -u\times\nu,
    \qquad 
    \dot{\beta} = 0 
\end{align*}
and these equations are of Lie-Poisson form:
indeed, 
$u = \delta h_C/\delta \nu$ 
and 
$h_C(\nu,\beta) 
= \vv<(\nu,\beta)^{\top}, (\mu_C^{\mathcal{P}})^{-1}(\nu,\beta)^{\top}>/2$.
\end{remark}

\subsection{Equivalence to Euler-Poincar\'e  matching conditions (\cite{BLM01b})}\label{sec:EPcond}
When $\G$ and $\D$ are finite dimensional, so that \eqref{LP1} and \eqref{LP2} can be expressed in local coordinated then these conditions coincide with the matching conditions (EP1) and (EP2) of \cite{BLM01b}. 

To see this equivalence, write the metric $\mu_0^{\mathcal{P}}$ in local coordinates $w_i$ and $Z_a$ on $\du$ and $\gu$, respectively,
such that $u = u^i w_i$ and $X = X^a Z_a$:
\begin{equation}
    \label{eS1:mu_loc} 
    \mu_0^{\mathcal{P}}
    = 
    \left(
    \begin{matrix}
     (\mu_0^{\mathcal{P}})_{ij} & (\mu_0^{\mathcal{P}})_{ia} \\ 
     (\mu_0^{\mathcal{P}})_{ai} & (\mu_0^{\mathcal{P}})_{ab}
    \end{matrix}
    \right) 
\end{equation}
 The index convention is to let $i,j,k = 1,\dots,\textup{dim}\,\du$ and $a,b,c = 1,\dots,\textup{dim}\,\gu$, and summation over repeated indices is implied.
The relation between $\mu_0$, $\I_0$ and $\A_0$, as appearing in \eqref{eS1:KK}, and the local coordinate data are
$(\mu_0^{\mathcal{P}})_{ij} 
= (\mu_0)_{ij} + (A_0)^c_i (\mu_0^{\mathcal{P}})_{cb} (\A_0)^b_j$, 
$(\I_0)_{ab} = (\mu_0^{\mathcal{P}})_{ab}$ and $(\A_0)^a_i = (\mu_0^{\mathcal{P}})^{ab}(\mu_0^{\mathcal{P}})_{bi}$. 

The controlled Lagrangian, $l_{\tau,\sigma,\rho}$, constructed by \cite{BLM01b}  is defined in terms of $\tau = \A_C - \A_0: \du\to\gu$ with coordinate expression $\tau^a_i$, and matrices $(\mu_{\rho}^{\mathcal{P}}) := \textup{diag}(0_{ij}, \rho_{ab})$ and  $(\mu_{\sigma}^{\mathcal{P}}) := \textup{diag}(0_{ij}, \sigma_{ab})$ where $\textup{diag}$ denotes the diagonal matrix, $(0_{ij})$ is the quadratic zero matrix on $\du$, and $(\rho_{ab})$, $(\sigma_{ab})$ are symmetric and invertible matrices.

The vertical space of the left (or right) $\G$-action  on $\P = \D\times\G$ is $\textup{VER} = \P\times\gu$. The horizontal space associated to $\A_0$ is $\textup{HOR}_0 = \P\times\{(u, -\A_0u)\}$ while the horizontal space associated to $\A_C$ is $\textup{HOR}_{\tau} = \P\times\{(u,-(\A_0+\tau)u)\}$. Therefore, the controlled Lagrangian (\cite[Equ.~(5)]{BLM01b}) is
\begin{align*} 
 l_{\tau,\sigma,\rho} \left(
   \begin{matrix}
    u \\ 
    X
   \end{matrix}
  \right) 
 &= 
 \frac{1}{2}\left\langle 
  \mu_0^{\mathcal{P}}
  \left(
   \begin{matrix}
    u \\ 
    -\A_0 u
   \end{matrix}
  \right) ,
    \left(
   \begin{matrix}
    u \\ 
    -\A_0 u
   \end{matrix}
  \right) 
 \right\rangle 
 + 
  \frac{1}{2}\left\langle 
  \mu_{\sigma}^{\mathcal{P}}
  \left(
   \begin{matrix}
    0 \\ 
    -\tau u
   \end{matrix}
  \right) ,
    \left(
   \begin{matrix}
    0 \\ 
    -\tau u
   \end{matrix}
  \right) 
 \right\rangle  \\
 &\phantom{==}
 + 
   \frac{1}{2}\left\langle 
  \mu_{\rho}^{\mathcal{P}}
  \left(
   \begin{matrix}
    0 \\ 
    X + (\A_0+\tau)u
   \end{matrix}
  \right) ,
    \left(
   \begin{matrix}
    0 \\ 
    X + (\A_0+\tau)u
   \end{matrix}
  \right) 
 \right\rangle \\
 &= 
 \frac{1}{2}\Big( 
  (\mu_0)_{ij} u^i u^j  +  \sigma_{ab}\tau^a_i\tau^b_j u^i u^j 
  + \rho_{ab} (\A_0+\tau)^a_i (\A_0+\tau)^b_j u^i u^j 
  + 2 \rho_{ab} (\A_0+\tau)^b_i u^i X^a + \rho_{ab}X^aX^b  
 \Big) \\
 &= 
    \frac{1}{2}\left\langle 
  \left( 
  \begin{matrix}
   (\mu_0)_{ij} + \sigma_{cd}\tau^c_i\tau^d_j + \rho_{cd} (\A_0+\tau)^c_i (\A_0+\tau)^d_j 
   & \rho_{bc} (\A_0+\tau)^c_i  \\ 
   \rho_{ac} (\A_0+\tau)^c_j  &   \rho_{ab}
  \end{matrix}
  \right) 
  \left(
   \begin{matrix}
     u^j \\ 
     X^b
   \end{matrix}
  \right) ,
    \left(
   \begin{matrix}
     u^i\\ 
     X^a 
   \end{matrix}
  \right) 
 \right\rangle \\
 &= 
    \frac{1}{2}\left\langle 
  \mu_C^{\mathcal{P}} 
  \left(
   \begin{matrix}
     u \\ 
     X
   \end{matrix}
  \right) ,
    \left(
   \begin{matrix}
     u\\ 
     X 
   \end{matrix}
  \right) 
 \right\rangle
\end{align*} 
The last equality defines $\mu_C^{\mathcal{P}}$ in matrix form. 
Because of \eqref{eS1:KK} this implies the coordinate expressions
\begin{align}
\label{eS1:muCloc}
    (\mu_C)_{ij} + (\A_C^*\I_C\A_C)_{ij} 
    &= 
    (\mu_0)_{ij} + \sigma_{cd}\tau^c_i\tau^d_j + \rho_{cd} (\A_0+\tau)^c_i (\A_0+\tau)^d_j \\
    \notag 
    (\I_C)_{ab}
    &= \rho_{ab} \\ 
    \notag 
    (\A_C)_i^b 
    &= \rho^{ba}\rho_{ac}\Big( (\A_0)^c_i + \tau^c_i\Big) 
\end{align}
which allows to reformulate \eqref{LP1} as
\begin{equation}
    \label{LP1loc}\tag{LP1loc} 
    \tau^a_i
    = \Big(\rho^{ab} - (\mu_0^{\mathcal{P}})^{ab}\Big) (\mu_0^{\mathcal{P}})_{bi}
\end{equation}
Using this condition and the above equation \eqref{eS1:muCloc} for  $(\mu_C)_{ij} + (\A_C^*\I_C\A_C)_{ij}$ it follows, after some calculation, that \eqref{LP2} is equivalent to 
\begin{equation}
    \label{LP2loc}\tag{LP2loc} 
    \sigma^{ab} 
    =  (\mu_0^{\mathcal{P}})^{ab} - \rho^{ab}. 
\end{equation}
Equations \eqref{LP1loc} and \eqref{LP2loc} coincide with the Euler-Poincar\'e matching conditions EP1,
$\tau_i^a = -\sigma^{ab} (\mu_0^{\mathcal{P}})_{bi}$, and EP2, 
$\sigma^{ab} +  \rho^{ab} =  (\mu_0^{\mathcal{P}})^{ab}$, of \cite{BLM01b}. The Euler-Poincar\'e formulation  is carried out in the Lagrangian setting while the Lie-Poisson approach is a special case of Hamiltonian mechanics. The equivalence of controlled Lagrangians and controlled Hamiltonians is analyzed systematically from a general point of view in \cite{CBLMW02}. 

In the finite dimensional case, Proposition~\ref{prop:matching} is thus equivalent to the Euler-Poncar\'e matching construction of \cite{BLM01b}.  
Proposition~\ref{prop:matching} is stated in a coordinate free manner and is shown to apply, together with Proposition~\ref{prop:Cons}, also in the non-Abelian setting. Hence this construction can be used in the Hall MHD example of Section~\ref{Sec3}, where the symmetry group $\G$ is non-Abelian and infinite dimensional. Moreover, since the theory of controlled (simple) Lagrangians is equivalent to that of controlled (simple) Hamiltonians (\cite{CBLMW02,CM04})  also the Euler-Poincar\'e matching conditions generalize to the non-Abelian setting and the corresponding controlled equations of motion are given by the Lagrangian analogue of \eqref{e:ceom1}.

\subsection{Controlled dissipation}\label{sec:ConDis}
Let $h_0(\nu,\alpha) = \vv<(\nu,\alpha)^{\top}, (\mu_0^{\mathcal{P}})^{-1}(\nu,\alpha)^{\top}>/2$ be of Kaluza-Klein type with $\mu_0^{\mathcal{P}}$ defined by \eqref{eS1:KK} in terms of $\mu_0$, $\I_0$ and $\A_0$. 
Consider maps $C: \du^*\to\gu^*$, $N:\du^*\to\gu^*$ and $G: \gu^*\to\gu^*$ and define $j$ as in definition~\eqref{e:def_j}. 
It is assumed that 
\begin{enumerate}[\up (1)]
    \item 
    $C$ is linear and  $\mu_C := (1+\A_0^*C)^{-1}\mu_0$ defines an inner product on $\du$;
    \item 
    $G$ is a (not necessarily equivariant) isomorphism.
\end{enumerate}
As above, let $\beta = G^{-1}(\alpha + C\nu + N(\nu))$ such that there is a diffeomorphism 
\[ 
 \du^*\times\gu^*\to\du^*\times\gu^*,\qquad
 (\nu,\alpha)\mapsto(\nu,\beta)
\] 
where $N$ may be nonlinear. 

It is not assumed that $C$ and $G$ satisfy the matching conditions of Proposition~\ref{prop:matching}. In fact, the assertions in this subsection hold with $G=1$, while $\A_C$ and $\I_C$ do not enter the analysis.  

The controlled equations~\eqref{eS1:ceom1LP}-\eqref{eS1:ceom2LP} are equivalent to the closed-loop equations~\eqref{e:ceom2}-\eqref{e:ceom3}. 
Subsequently, the goal is to asymptotically stabilize the equations at an equilibrium $(\nu_e,\alpha_e)$ of \eqref{eS1:ceom1LP}-\eqref{eS1:ceom2LP} such that the closed-loop relation $G\beta_e = \alpha_e+C\nu_e+N(\nu_e)$ is preserved. 
To this end a dissipative control, 
$\tilde{\mathcal{U}}_{\textup{diss}}(\nu,\alpha) 
 = G\,\mathcal{U}_{\textup{diss}}(\nu,\beta)
$,
is now added to \eqref{eS1:ceom2LP}. 
Thus, 
\begin{equation}
    \label{eS1:LPdiss}
    \dot{\nu} 
    = \pm\ad(\delta h_0/\delta\nu)^*\nu ,
    \qquad 
    \dot{\alpha}
    = \pm\ad(\delta h_0/\delta\alpha)^*\alpha + \mathcal{U}_{\textup{LP}}(\nu,\alpha) + \tilde{\mathcal{U}}_{\textup{diss}}(\nu,\alpha)
\end{equation}
such that we obtain the closed loop system 
\begin{align}
  \label{eS1:LPdiss1}
    \dot{\nu} 
    &= 
    \pm\,\ad(\mu_C^{-1}\nu)^*\nu 
    \,\mp\,\ad(\mu_0^{-1}\A_0^* G\beta)^*\nu 
    \,\pm\,\ad(\mu_0^{-1}\A_0^* N(\nu))^*\nu  \\
      \label{eS1:LPdiss2}
    \dot{\beta}
    &= \pm\,\ad(\delta h_0/\delta\alpha)^*\beta 
    + \mathcal{U}_{\textup{diss}}(\nu,\beta)
\end{align}
where $\delta h_0/\delta\alpha = (\I_0^{-1}+\A_0\mu_0^{-1}\A_0^*)(G\beta - C\nu - N(\nu)) - \A_0\mu_0^{-1}\nu$. 

Let $s$ be a parameter which will be chosen $+1$ or $-1$, depending on whether the goal is to stabilize the system at an energy minimum or maximum. Define
\begin{equation}
    \label{eS1:defN}
    N(\nu) 
    := \pm\,s\I_0 \A_0 \mu_0^{-1} \ad(\mu_C^{-1}\nu)^*\nu 
\end{equation}
and 
\begin{equation}
    \label{eS1:defUdiss} 
    \mathcal{U}_{\textup{diss}}(\nu,\beta) 
    := 
    \mp\,\ad(\delta h_0/\delta\alpha)^*\beta 
    - \beta 
    - G^{-1} N(\nu) 
\end{equation}
Consider now the function $g_C: \du^*\times\gu^*\to\R$ defined by 
\begin{equation}
    \label{eS1:def_g_C}
    g_C(\nu,\beta) 
    := \frac{1}{2}\vv<\nu, \mu_C^{-1}\nu> + \frac{s}{2}\vv<G\beta, \I_0^{-1}G\beta>
\end{equation}
This function is not the Hamiltonian of the controlled system \eqref{eS1:LPdiss1}-\eqref{eS1:LPdiss2}, and does not satisfy the matching conditions of Proposition~\ref{prop:matching} (except in the trivial case $\A_0=0$). The motivation for $g_C$ is that at $\beta=0$ it coincides with the Hamiltonian function corresponding to the (reduced) equation $\dot{\nu} = \pm\ad(\mu_C^{-1}\nu)^*\nu$, while in $\beta$ it is the most natural quadratic choice. The dissipative force, $\mathcal{U}_{\textup{diss}}$, satisfies
\begin{align}
    \label{eS1:g_C}
    \frac{\del}{\del t}g_C(\nu_t,\beta_t) 
    &= 
    \mp\,\Big\langle \ad\Big(\mu_0^{-1}\A_0^*G\beta)\Big)^*\nu, \mu_C^{-1}\nu\Big\rangle
    \,\pm\, 
    \Big\langle \ad\Big(\mu_0^{-1}\A_0^* N(\nu) \Big)^*\nu, \mu_C^{-1}\nu\Big\rangle  \\ 
    &\phantom{==}
    + 
    s\Big\langle  
     -G\beta - N(\nu), \I_0^{-1}G\beta 
    \Big\rangle 
    \notag \\
    \notag 
    &= 
    -s \Big\langle 
     \ad(\mu_C^{-1}\nu)^*\nu , \mu_0^{-1}\A_0^*\I_0\A_0\mu_0^{-1}\,\ad(\mu_C^{-1}\nu)^*\nu
    \Big\rangle 
    - s \Big\langle G\beta, \I_0^{-1}G\beta \Big\rangle 
    \\ 
    \notag 
    &=
    - s\, \Big|\Big| N(\nu) \Big|\Big|_{\mathbb{I}_0^{-1}}^2
    -s \, \Big|\Big| G\beta \Big|\Big|_{\mathbb{I}_0^{-1}}^2
\end{align}
where $||\cdot||_{\mathbb{I}_0^{-1}}$ is the norm associated to $\I_0^{-1}$ on $\gu^*$. 
Hence, for $s = 1$, the closed loop system is weakly dissipative: 
\[
  \frac{\del}{\del t}\,g_C(\nu,\beta) \le 0 
\]
Reversing the sign, $s=-1$, yields $\del g_C(\nu)/\del t \ge 0$.

If $K: \du^*\to\R$ is a Casimir function (i.e., constant on coadjoint orbits) equation~\eqref{eS1:LPdiss1} implies that $\del K(\nu_t)/\del t = 0$. 

\begin{remark}[Stability]\label{rem:stable} 
Suppose that $\nu_e$ is an unstable equilibrium of the uncontrolled equation $\dot{\nu} = \pm\ad(\mu_0^{-1}\nu)^*\nu$. Suppose  further that $\ad(\mu_C^{-1}\nu_e)^*\nu_e$ and consequently $N(\nu_e) = 0$.  

The question of asymptotic stability of the equilibrium $(\nu_e, \alpha_e = -C\nu_e)$ of the controlled set of equations \eqref{eS1:LPdiss} may be approached in the following manner:
\begin{enumerate}[\up (1)]
    \item 
    Observe that $(\nu_e,\beta_e=0)$ is an equilibrium of \eqref{eS1:LPdiss1}-\eqref{eS1:LPdiss2} by construction. 
    \item 
    Find a Casimir function $K_C: \du^*\to\R$ such that the first variation at $\nu_e$ satisfies $D^1K_C(\nu_e) = -\mu_C^{-1}\nu_e$. It follows that $D^1(g_C+K_C)(\nu_e,0)(\nu,\beta) = \vv<\nu,\mu_C^{-1}\nu_e> + D^1K_C(\nu_e)(\nu) + s\vv<G\beta,\I_0^{-1}G0> = 0$.
    \item
    Ensure that  $C$ is such that the second variation $D^2(g_C+K)(\nu_e,0)$ is positive definite (with $s=1$) or negative definite (with $s=-1$). 
    \item 
    Conclude that $s(g_C+K_C)(\nu_t,\beta_t) - s(g_C+K_C)(\nu_e,0)$ is a decreasing Lyapunov function, and use LaSalle's invariance principle if the inequality $s\del(g_C+K)/\del t \le 0$ is not strict in a punctured neighbourhood of $(\nu_e,0)$.
\end{enumerate}
In the infinite dimensional case additional convexity arguments  may be needed (\cite{HMRW85}). 
\end{remark}


\begin{remark}[Double bracket dissipation]
Let $\mathcal{U}_{\textup{diss}} = 0$. Assume $\beta_0 = 0$, initially, which is  preserved by the dynamics \eqref{eS1:LPdiss2} so that $\beta=0$ as long as solutions exist. Define $h_C = h_C(\nu) = \vv<\nu, \mu_C^{-1}\nu>/2$. In line with Remark~\ref{rem:beta0}, the quantities $\I_C$ and $\A_C$ do not enter the analysis.  

With the above definition of $N(\nu)$, equation~\eqref{eS1:LPdiss1} is of double bracket form (\cite{BKMR96}):
\begin{equation}
    \label{e:double_bra}
    \dot{f}
    = \{f,h\} - \{\{f,h\}\}
\end{equation}
where $\{.,.\}$ is a Poisson bracket and $\{\{.,.\}\}$ is a symmetric bracket. This reformulation follows from the definitions $\Gamma := \pm \mu_0^{-1}\A_0^*\I_0\A_0\mu_0^{-1}: \du^*\to\du$ and
\begin{align*}
     \{f,h\}(\nu) 
     &= \mp\Big\langle\nu, \ad\Big(\frac{\delta f}{\delta\nu}\Big)\frac{\delta h}{\delta\nu} \Big\rangle \\
     \{\{f,h\}\}(\nu)
     &= \Big\langle\ad\Big(\frac{\delta f}{\delta\nu}\Big)^*\nu, 
        \Gamma\,\ad\Big(\frac{\delta h}{\delta\nu}\Big)^*\nu \Big\rangle
\end{align*}
\end{remark}

\begin{remark}[Non-matching]
The control maps $C$ and $G$ in \eqref{eS1:LPdiss1}-\eqref{eS1:LPdiss2} need not satisfy the matching conditions of Proposition~\ref{prop:matching}. However, this does not mean that no use of Hamiltonian structure is made. In the $\nu$ variable it is used that coadjoint orbits are preserved, and a Casimir function $K_C = K_C(\nu)$ is necessary. On the other hand, in the $\beta$ variable the goal is to asymptotically approach the origin, and hence a violation of the coadjoint orbit structure is desirable.  
\end{remark}

\subsection{Port controlled Hamiltonian systems: interconnection and damping assignment}
Let $\Pi_0: T^*\po^*\to T\po^*$ denote the Poisson tensor associated to the (plus or minus) Lie-Poisson structure on $\po^*=\du^*\times\gu^*$, and $\iota: \gu^*\to \du^*\times\gu^*$, $\alpha\mapsto(0,\alpha)$ the inclusion of the second factor.  
The second dual $T^*\po^* = \po^*\times\po^{**}$ is identified with $\po^*\times\po$. In the infinite dimensional case this means that only the regular part of the dual is considered. Hence $\Pi_0(x): \po\to\po^*$. 
Then equation~\eqref{eS1:LPdiss} may be reformulated as a port controlled Hamiltonian system  on $\po^*$:
\begin{align}
    \label{eS1:pch}
    \dot{x}
    &= \Pi_0(x)\,\frac{\delta h_0}{\delta x} + \iota\, \mathcal{U} \\
    \notag 
    y &= \delta h_0/\delta\alpha 
\end{align}
where $x=(\nu,\alpha)$, $\delta h_0/\delta x = (\delta h_0/\delta\nu,\delta h_0/\delta\alpha)$, and $y$ and $\mathcal{U} = \mathcal{U}_{\textup{LP}} + \mathcal{U}_{\textup{diss}}$ are conjugate in the sense that $\del h_0/\del t = \vv<\mathcal{U},\delta h_0/\delta\alpha>$.

The objective of interconnection and damping assignment -- passivity based control (IDA-PBC) is now to find a desired structure, consisting of a new skew symmetric tensor  $\Pi_d(x): \po\to\po^*$, a positive and symmetric tensor  $\mathcal{R}_d(x): \po\to\po^*$ and a new function $h_d$, so that \eqref{eS1:pch} may be expressed as 
\begin{equation}
    \label{eS1:pch2}
    \dot{x} 
    = \Big( \Pi_d(x) - \mathcal{R}_d(x) \Big) \, \frac{\delta h_d}{\delta x}. 
\end{equation}
The bracket associated to $\Pi_d$, namely $\{f,g\}_d = \vv<\Pi_d \,\delta f/\delta x, \delta g/\delta x >$, is bilinear, skew-symmetric and satisfies the Leibniz identity, but not necessarily the Jacobi identity. It is therefore an almost Poisson structure. See \cite{OSME02, OG04}.

With \eqref{eS1:defN} and \eqref{eS1:defUdiss} we may reformulate the closed loop equations~\eqref{eS1:LPdiss1} and \eqref{eS1:LPdiss2} as 
\begin{align}
    \label{eS1:pH1} 
    \dot{\nu} 
    &=
    \pm\,\ad\Big(\frac{ g_C}{\delta\nu} \Big)^*\nu 
    \,\mp\,s\,\ad\Big(\mu_0^{-1}\A_0^* G \I_0\frac{\delta g_C}{\delta\beta} \Big)^*\nu 
    + s\,\ad\Big(\mu_0^{-1}\A_0^* \I_0 \A_0 \mu_0^{-1} \ad(\frac{ g_C}{\delta\nu})^*\nu \Big)^*\nu  \\
      \label{eS1:pH2}
    \dot{\beta}
    &=  
    - s\I_0\frac{\delta g_C}{\delta\beta}
    \,\mp\, s \,G^{-1}\I_0  \A_0 \mu_0^{-1} \ad(\frac{ g_C}{\delta\nu})^*\nu 
\end{align}
where, as before, $s$ is $+1$ or $-1$, and we use that $\delta g_C/\delta\nu = \mu_C^{-1}\nu$ and $\delta g_C/\delta\beta = s\I_0^{-1}\beta$. This system involves only first order derivatives of $g_C$, hence one can algebraically arrive at the form \eqref{eS1:pch2}. 

Let $z = (\nu,\beta) = (\nu, G^{-1}(\alpha+C\nu+N(\nu))) =: \Phi(\nu,\alpha) = \Phi(x)$ and define 
\begin{equation}
    \Pi_C(z)(v,Y) := 
    \left(\begin{matrix}
        \pm\ad(v)^*\nu \,\mp\, s\,\ad\Big(\mu_0^{-1}\A_0^*\I_0(G^*)^{-1}Y \Big)^*\nu \\
        \mp\,s\, G^{-1}\I_0\A_0\mu_0^{-1}\ad(v)^*\nu 
    \end{matrix}\right)     
\end{equation}
which is skew symmetric, and 
\begin{equation}
    \mathcal{R}_C(z)(v,Y)
    := 
    s\left(\begin{matrix}
     -\ad\Big(\mu_0^{-1}\A_0^*\I_0\A_0\mu_0^{-1}\ad(v)^*\nu \Big)^*\nu  \\
     G^{-1}\I_0(G^*)^{-1} Y 
    \end{matrix}\right) 
\end{equation}
which is symmetric and positive (resp., negative) semi-definite for $s=1$ (resp., $s=-1$). 
As above, $\Pi_C$ and $\mathcal{R}_C$ are viewed as bundle homomorphisms $\po^*\times\po = T^*\po^*\to T\po^* = \po^*\times\po^*$, although both are independent of $\beta$.

With this notation it follows that \eqref{eS1:pH1}-\eqref{eS1:pH2} can be written in the desired form  
\begin{equation}
\label{eS1:pHS}
    \dot{z} 
    = \Big( \Pi_C(z) - \mathcal{R}_C(z)\Big)\frac{\delta g_C}{\delta z} . 
\end{equation}
This equation is expressed in the $z$-coordinates. To obtain \eqref{eS1:pch2}, we can use the pullback along $\Phi: x = (\nu,\alpha) \mapsto z = (\nu,\beta) = (\nu, G^{-1}(\alpha + C\nu + N(\nu)))$. Indeed, the  new structure is given by the corresponding pullbacks $\Pi_d = \Phi^*\Pi_C$, $\mathcal{R}_d = \Phi^*\mathcal{R}_C$ and $h_d = \Phi^*g_C = g_C\circ\Phi$. Using \eqref{eS1:pullb} we have from \eqref{eS1:pHS} that 
\begin{align}
\label{eS1:pHSx}
    \dot{x}
    &= T\Phi^{-1}\,\dot{z} 
    = T\Phi^{-1}\,\Big( \Pi_C(\Phi(x)) - \mathcal{R}_C(\Phi(x))\Big)
        \Big(T\Phi^{-1}\Big)^*\Big(T\Phi\Big)^*
        \frac{\delta g_C}{\delta \Phi(x)} \\
     \notag    
    &= \Big( \Pi_d(x) - \mathcal{R}_d(x)\Big)\frac{\delta h_d}{\delta x} . 
\end{align}

\begin{remark}
The new Hamiltonian function, $h_d(\nu,\alpha) = \vv<\nu,\mu_C^{-1}\nu>/2 + s\vv<\alpha+C\nu+N(\nu), \I_0^{-1}(\alpha+C\nu+N(\nu))>/2$, is fourth order in $\nu$ owing to the nonlinearity in $N(\nu)$. Similarly, $\mathcal{R}_d$ is nonlinear in $\nu$.
The bracket associated to $\Pi_d$ does in general, unless $\A_0=0$ or when $\du$ is Abelian,  not satisfy the Jacobi identity.     
\end{remark}

\begin{theorem}\label{thm:FB_IDA}
Let $h_0$ be the kinetic energy Hamiltonian associated to \eqref{eS1:KK}.
Let $\Phi(\nu,\alpha)=(\nu,G^{-1}(\alpha+C\nu+N(\nu)))$ where $G: \gu^*\to\gu^*$ is an isomorphism, $C: \du^*\to\gu^*$ is linear such that $\mu_C := (1+\A_0^*C)^{-1}\mu_0$ defines an inner product on $\du$, and $N: \du^*\to\gu^*$ is given by \eqref{eS1:defN}. 
Then the following equations of motion are equivalent: 
\begin{enumerate}[\up (1)]
\item 
The feedback controlled system defined by \eqref{eS1:LPdiss} with
$\delta h_0/\delta\nu = \mu_0^{-1}\nu - \mu_0^{-1}\A_0^*\alpha$: 
\begin{align}\label{eS1:LPdissXPL}
    \dot{\nu}
    &= \pm\ad(\delta h_0/\delta\nu)^*\nu \\
    \notag
    \dot{\alpha}
    &= 
    \mp
         C\, \ad(\delta h_0/\delta\nu)^*\nu 
    -
     s\I_0\A_0\mu_0^{-1} \Big(
      \ad( \mu_C^{-1}\ad(\delta h_0/\delta\nu)^*\nu )^*\nu
      + \ad( \mu_C^{-1}\nu )^*(\ad(\delta h_0/\delta\nu)^*\nu \,\pm\,2\nu)
     \Big)
    -  \alpha - C\nu 
\end{align}
\item 
The IDA system \eqref{eS1:pHSx} with $x=(\nu,\alpha)$ and the fourth order Hamiltonian $h_d$. 
\item 
The IDA system \eqref{eS1:pHS} with  $z=\Phi(x)$ and the quadratic Hamiltonian $g_C$.
\end{enumerate}
\end{theorem}

\begin{proof}
It only remains to note  that equation~\eqref{eS1:LPdissXPL} follows from \eqref{eS1:LPdiss}, \eqref{eS1:tildeU}, \eqref{eS1:defUdiss}, and since \eqref{eS1:defN} implies
\[
 dN(\nu)\,\eta 
 = \pm s\I_0\A_0\mu_0^{-1}
  \Big(
      \ad( \mu_C^{-1}\eta )^*\nu
      + \ad( \mu_C^{-1}\nu )^*\eta
     \Big)
\]
for $\nu,\eta\in\du^*$.
\end{proof}

In particular, the stability analysis of an equilibrium $x_e = (\nu_e,\alpha_e)$ of \eqref{eS1:LPdissXPL} can be carried out equivalently for the equilibrium $z_e = \Phi(x_e) = (\nu_e,\beta_e)$ of the IDA system \eqref{eS1:pHS}. If $\beta_e=0$ then stability, or asymptotic stability, of $z_e$ can be addressed according to Remark~\ref{rem:stable}. 

The map $G$ can be set to $G=1$. Letting $C = k\I_0\A_0\mu_0^{-1}$ renders $\mu_C$ symmetric and invertible for small $k\in\R$.

\section{Satellite with rotor}\label{sec2:sat}

The stabilization of a satellite by means of an internal rotor, modelled by a carrier rigid body with a wheel attached to one of the principal axes, has been treated systematically for the first time in \cite{K85} and later in \cite{WK92,BKMS92,BMS97,BLM01b,BCLMW00}. In these references the rotor is attached to the long axis, and a feedback law for the rotor's speed relative to the axis is constructed such that rotation of the carrier rigid body about the intermediate axis becomes a nonlinearly stable equilibrium. Further, \cite{BCLMW00} construct an additional control which acts dissipatively such that asymptotic stability is achieved. 

In this section we consider the satellite with a rotor attached to the short axis and apply the double bracket IDA construction of Theorem~\ref{thm:FB_IDA} to obtain asymptotic stability of rotation about the middle axis in the full phase space of $(\nu,\alpha)$-variables corresponding to rigid body and rotor angular momenta. 

Choosing the short axis as a base for the attached rotor does not lead to any significant changes since the two cases, long or short, are completely analogous. 
Attaching the rotor to the short (respectively: long) axis leads to an energy minimum (maximum) at the intermediate axis for the controlled dynamics. Thus  asymptotic stability is found by defining the symmetric bracket with the correct sign so that energy either decreases towards the minimum ($s=1$) or increases towards the maximum ($s=-1$). The double bracket dissipation used here is different from the dissipative control considered in \cite{BCLMW00}.

\subsection{Equations of motion for the free system}
The configuration space of the carrier rigid body  is $\D=\SO(3)$ and that of the rotor attached to the short axis is the Abelian Lie group $\G=S^1$. The corresponding Lie algebras are $\du=\R^3$ equipped with the cross product $\ad(u)v = [u,v] = u\times v$ and $\gu=\R$. The total space is denoted, as above, by $\P = \D\times\G$. To trivialize the tangent bundle we choose the left (also called: body) representation, $T\P=\P\times\po = \D\times\G\times\du\times\gu$ via left multiplication, and accordingly for the cotangent bundle,  $T^*\P=\P\times\po^* = \D\times\G\times\du^*\times\gu^*$, where $\po^*$, $\du^*=\R^3$ and $\gu^*=\R$ are the corresponding dual Lie algebras. This means that the top positioned sign in formulas in Section~\ref{sec1:feedback} is the relevant one, i.e.\ $\pm$ becomes $+$. 
 
The carrier body principal moments of inertia are denoted by $I_1<I_2<I_3$ and those of the rotor are $i_1=i_2<i_3$ such that $\lam_1<\lam_2<I_3$ where $\lam_j=I_j+i_j$ for $j=1,2,3$. 

The Hamiltonian of the system is given by the kinetic energy $h_0: \po^*=\R^3\times\R\to\R$, $h_0(\nu,\alpha) = \vv<(\nu,\alpha)^{\top},(\mu_0^{\mathcal{P}})^{-1}(\nu,\alpha)^{\top}>/2$ where $(\cdot)^{\top}$ denotes vector transpose and $\vv<.,.>$ is the standard scalar product. Let $\mu_0 := \diag{\lam_1,\lam_2,I_3}$ be the diagonal matrix with these entries, $\I_0 := i_3$, and $\A_0 := e_3^{\top}: \du=\R^3\to\gu=\R$, $\A_0(u)=u_3$ be the projection onto the third coordinate. Then the kinetic energy metric $\mu_0^{\mathcal{P}}$ can be expressed in Kaluza-Klein form: 
\begin{equation}
    \label{eS2:KKmetric}
    \mu_0^{\mathcal{P}}
    = 
    \left(
    \begin{matrix}
     \lam_1 & & &  \\
      & \lam_2 & &  \\ 
      & & \lam_3 & i_3 \\
      & & i_3 & i_3 
    \end{matrix}
    \right)
    = 
    \left(
    \begin{matrix}
     \mu_0+\A_0^*\I_0\A_0  &  \A_0^*\I_0 \\
     \I_0\A_0 & \I_0 
    \end{matrix}
    \right)
    , \qquad
    (\mu_0^{\mathcal{P}})^{-1}
    = 
    \left(
    \begin{matrix}
     \mu_0^{-1}  &  -\mu_0^{-1}\A_0^* \\
     -\A_0\mu_0^{-1} & \I_0^{-1} + \A_0\mu_0^{-1}\A_0^* 
    \end{matrix}
    \right)
\end{equation}
where $\A_0^* = (e_3^{\top})^{\top} = e_3$. The free equations of motion, i.e.\ without control force, are
\begin{align}
\label{eS2:eom_fr}
 \dot{\nu} &= \ad(u)^*\nu = -u\times\nu,\qquad 
 \dot{\alpha} 
 = 0
\end{align}
where $u = \delta h_0/\delta\nu = \mu_0^{-1}(\nu-\A_0^*\alpha) = (\lam_1^{-1}\nu^1, \lam_2^{-1}\nu^2, I_3^{-1}\nu^3 - \alpha)^{\top}$.

The equilibrium $(\nu_e, \alpha_e) = (e_2,0)$ where $e_2 = (0,1,0)^{\top}$ corresponds to rotation about the middle axis, and this is an unstable equilibrium for \eqref{eS2:eom_fr}.

\subsection{Feedback control}
Following \eqref{eS1:LPdissXPL}, the controlled equations of motion are 
\begin{align}
\label{eS2:eomF1}
 \dot{\nu} &= \ad(u)^*\nu = -u\times\nu\\
 \notag 
 \dot{\alpha} 
 &= \mathcal{U}_{\textup{LP}}(\nu,\alpha) + \widetilde{\mathcal{U}}_{\textup{diss}}(\nu,\alpha) \\
 \notag 
 \mathcal{U}_{\textup{LP}}(\nu,\alpha)  
 &= (C+ dN(\nu))\,u\times\nu    \\
  \notag
  \widetilde{\mathcal{U}}_{\textup{diss}}(\nu,\alpha)
  &= -\alpha - C\nu - 2N(\nu)
\end{align}
where $C: \R^3\to\R$ is linear and $N: \R^3\to\R$ is defined by \eqref{eS1:defN}. Further, $\mu_C = (1+\A_0^*C)^{-1}\mu_0$ should be a symmetric and positive matrix on $\R^3$. Thus $C$ has to be a multiple of $\A_0$,  $C := k\I_0\A_0\mu_0^{-1} = k e_3^{\top} i_3/I_3$ with $k i_3/I_3>-1$, whence 
\begin{equation}
\label{eS2:muC}
    \mu_C
    := \diag{\lam_1,\lam_2,(1+ki_3/I_3)^{-1}I_3}.
\end{equation} 

\begin{remark}[Lie-Poisson matching]
If we consider the case $N=0$ and $\widetilde{\mathcal{U}}_{\textup{diss}}=0$ such that the matching conditions of Proposition~\ref{prop:matching} can be applied then the controlled equations of motion are Lie-Poisson in the closed loop variables $(\nu,\beta) = (\nu, G^{-1}(C\nu+\alpha))$. The corresponding structure consisting of a Kaluza-Klein type metric $\mu_C^{\mathcal{P}}$ and an associated kinetic energy Hamiltonian $h_C$  is provided in Remark~\ref{rem:sat3}.
\end{remark}

\subsection{Asymptotic stabilization via double bracket IDA-PBC}\label{sec:sat_AS}
Consider the feedback controlled system \eqref{eS2:eomF1} with $N$ defined by \eqref{eS1:defN}, that is $N(\nu) = -s\I_0\A_0\mu_0^{-1}\,\mu_C^{-1}\nu\times\nu$ and $s=\pm1$ remains to be specified. We set
\[
 G=1, \qquad C = k\I_0\A_0\mu_0^{-1} = k\frac{i_3}{I_3}e_3^{\top}
\]
and define $\mu_C$ as above. 

The goal is to asymptotically stabilize the motion about the middle axis, i.e.\ to asymptotically stabilize the equilibrium $(\nu_e,\alpha_e) = (e_2,0)$ of the system \eqref{eS2:eomF1}, which is unstable for the free motion \eqref{eS2:eom_fr}. The idea is that $C$ should be used so that $e_2$ becomes the new short axis with respect to the controlled metric $\mu_C$; this holds if, and only if, $(1+ki_3/I_3)^{-1}I_3 < \lam_2$ since this makes $\lam_2$ the `new' largest moment of inertia. To see this systematically, and conclude asymptotic stability of the controlled motion, we use Theorem~\ref{thm:FB_IDA}. 

The analysis is carried out in the variables $(\nu,\beta) = (\nu,\alpha+C\nu+N(\nu)) = \Phi(\nu,\alpha)$. The system defined by \eqref{eS1:pHS} is given by 
\begin{align}
    \label{eS2:pHSsat}
    \dot{\nu}
    &= -\frac{\delta g_C}{\delta\nu}\times\nu 
       + s\Big(\mu_0^{-1}\A_0^*\I_0\frac{\delta g_C}{\delta\beta}\Big)\times\nu 
       + s\Big(
            \mu_0^{-1}\A_0^*\I_0\A_0\mu_0^{-1} (\frac{\delta g_C}{\delta\nu}\times\nu )
        \Big)\times\nu \\ 
\notag
    \dot{\beta}
    &= -s\I_0\frac{\delta g_C}{\delta\beta}
     + s \I_0\A_0\mu_0^{-1}\Big( \frac{\delta g_C}{\delta\nu}\times\nu\Big)
\end{align}
where 
\[
 g_C(\nu,\beta) 
 = \frac{1}{2}\vv<\nu,\mu_C^{-1}\nu> + \frac{s}{2}i_3 \beta^2
\]
Notice that $\Phi(\nu_e,\alpha_e) = (\nu_e,\beta_e)$ with $\beta_e=\alpha_e=0$ is  the equilibrium of \eqref{eS2:pHSsat}.
Casimir functions of $\Pi_C$, $\mathcal{R}_C$ are of the form $K_{\rho}(\nu,\beta) = \rho(|\nu|^2/2)$ for a smooth function $\rho: \R\to\R$. Candidate Lyapunov functions are formed following the energy-Casimir method as $g_C + K_{\rho}$. 

The condition for the first variation 
\[
 D^1(g_C+K_{\rho})(\nu_e,\beta_e)(\delta\nu,\delta\beta)
  = \vv<\lam_2^{-1}\nu_e, \delta\nu> + \rho'(1/2)\vv<\nu_e,\delta\nu>
\]
to vanish at the equilibrium is that $\rho'(1/2) = -\lam_2^{-1}$. The second variation at $(\nu_e,\beta_e)$ is 
\[ 
 D^2(g_C+K_{\rho})(\nu_e,\beta_e)(\delta\nu,\delta\beta)^2 
 = 
 \left
 <\left(\begin{matrix}
     (\lam_1^{-1}-\lam_2^{-1})\delta\nu^1 \\ 
     \rho''(1/2)\delta\nu^2 \\ 
     (I_3^{-1}(1+ki_3/I_3) - \lam_2^{-1})\delta\nu^3
 \end{matrix}\right>  , 
 \left(\begin{matrix}
     \delta\nu^1\\
     \delta\nu^2\\
     \delta\nu^3
 \end{matrix}\right) 
 \right)
 + s i_3^{-1}(\delta\beta)^2 .
\]
Since $\lam_1^{-1}-\lam_2^{-1}>0$, by assumption, the quadratic form  $D^2(g_C+K_{\rho})(\nu_e,\beta_e)$ cannot be negative definite. 
If $k=0$ it is indefinite since $I_3^{-1}-\lam_2^{-1}<0$. 
But it is positive definite if 
\begin{align}
\label{eS2:pd1}
  s&=1\\
\label{eS2:pd2}
  \rho''(1/2) &> 0 \\
\label{eS2:pd3}  
  k &> I_3\frac{I_3 - \lam_2}{i_3\lam_2}
\end{align} 
With these choices, $L(\nu,\beta) := g_C(\nu,\beta)+K_{\rho}(\nu)-g_C(\nu_e,\beta_e)-K_{\rho}(\nu_e)$ is a Lyapunov function for $(\nu_e,\beta_e)$ and \eqref{eS2:pHSsat}: it is strictly positive in a punctured neighborhood of $(\nu_e,\beta_e)$ and $\del L/\del t\le0$ along solutions due to \eqref{eS1:g_C}. Hence $(\nu_e,\beta_e)$ is a nonlinearly stable equilibrium of the dissipative system \eqref{eS2:pHSsat}, and the same holds for $(\nu_e,\alpha_e)$ with respect to the feedback controlled system \eqref{eS2:eomF1}. 

To conclude asymptotic stability, we use LaSalle's invariance principle and note that \eqref{eS1:g_C} implies that $\del L(\nu_t,\beta_t)/\del t = 0$ if, and only if, $N(\nu_t)=0$ and $\beta_t=0$. If $(\nu_t, \beta_t) = (\nu_e+\delta\nu_t,\delta\beta_t)$ is a perturbed solution this yields  $\delta\beta_t=0$.
Furthermore,
\[
 N(\nu_t)
 = \I_0\A_0\mu_0^{-1}\ad(\mu_C^{-1}\nu_t)^*\nu_t
 = -\frac{i_3}{I_3}(\lam_1^{-1}-\lam_2^{-1})\nu_t^1\nu_t^2
 = 0
\]
if, and only if, $\nu_t^1 = \delta\nu_t^1=0$ since $\nu_t^2$ remains close to $1$ by stability.  Since $(\nu_t, \beta_t) = (\nu_e+\delta\nu_t,\delta\beta_t)$ has  to be a solution of \eqref{eS2:pHSsat} the identities $\delta\beta_t=0$ and $\delta\nu_t^1=0$ imply that also $\delta\nu_t^3=0$. Hence by the LaSalle invariance principle $(\nu_t, \beta_t) = (\nu_e+\delta\nu_t,\delta\beta_t)$ tends asymptotically to $(\tilde{\nu}_e,\beta_e)$ with $\tilde{\nu_e} = |\nu_e+\delta\nu_0|e_2$. That is, the axis $\R e_2\times\{\beta_e\}$ is an asymptotically stable manifold for the closed-loop motion \eqref{eS2:pHSsat}. The same is therefore true for $\Phi^{-1}(\R e_2 \times \{\beta_e\}) = \R e_2\times\{\alpha_e\}$ and the controlled motion \eqref{eS2:eomF1}. When $|\nu_e+\delta\nu_0| \neq |\nu_e|$ the perturbed motion $\nu_e+\delta\nu_t$ cannot be expected to asymptotically approach $\nu_e$ since coadjoint orbits, i.e.\ spheres, are preserved in the $\nu$-variable.

\section{Feedback controlled Hall MHD flow}\label{Sec3}
Magnetohydrodynamics (MHD) describes the collective motion of a quasi-neutral plasma fluid such that the magnetic field lines are frozen in. Hall MHD allows for a decoupling between ions and magnetic field but assumes that the latter  is dragged along the electron fluid (\cite{Lighthill60,Holm87,KMT21,O07}).

\subsection{Control problem}\label{subs:problem}
Let $M\subset\R^3$ be a closed domain with smooth boundary, $\del M$. 
The two-fluid Euler-Maxwell system for inviscid and incompressible flow of ions and electrons, labeled by  $l=i,e$,  in the presence of an electromagnetic field $(E,B)$ is given by the Euler equations coupled to the Maxwell equations. That is, 
\begin{align}
\label{eS3:EulMax}
    m^l\frac{Du^l}{Dt}
    &= 
    -\nabla p^l + q^l(E + u^l\times B)\\ 
    \notag
    \textup{div}\,u^l &= 0 \\
    \notag
    \textup{div}\,B &= 0,
    \qquad B = \textup{curl}\, A, \qquad  \textup{div}\,A = 0 \\ 
    \notag
    \textup{div}\,E &= \sigma\\
    \notag
    \dot{A} &= -E \\ 
    \notag
    \textup{curl}\, B &= \mathcal{J}   
\end{align}
where $u^l$ is the species' fluid velocity, $m^l$ is the particle's mass, $p^l$ is the pressure determined by the incompressibility condition and $\frac{Du^l}{Dt} = \dot{u}^l + \vv<u,\nabla>u =  \dot{u}^l + \sum u^i\del_i u$ is the convective derivative. Further, $q^i = e$ and $q^e=-e$ is the electron charge, and $\mathcal{J}$ is the current density. These equations for the electric field $E$ and the magnetic field $B$ are the Maxwell equations without displacement term. We assume charge neutrality, $\sigma=0$, and that the current is a sum of induced and externally controlled components, i.e.
\begin{equation}
    \label{eS3:J}
    \mathcal{J} = eu^i - eu^e  + \mathcal{J}_{\textup{ext}}.
\end{equation}
The Hall MHD description is obtained from the two fluid system by neglecting the electron inertia, $m^e Du^e/Dt =0$. This yields Ohm's law with Hall term: 
\[
 E = -e^{-1}\nabla p^e - u^e\times B
 = 
 -e^{-1}\nabla p^e - \Big(u^i - e^{-1}\textup{curl}\,B + e^{-1}\mathcal{J}_{\textup{ext}}\Big)\times B
\]
where we have substituted $u^e = u^i - e^{-1}\textup{curl}\,B + e^{-1}\mathcal{J}_{\textup{ext}}$. Dropping the superscript $i$ and inserting this equation for $E$ in the system \eqref{eS3:EulMax} leads to the single fluid equations of motion for ion flow:
\begin{align}
    \label{eS3:EulMaxCo}
    \dot{u}
    &=  -\nabla p - \vv<u,\nabla>u + \textup{curl}\,B\times B + \mathcal{F}\\
    \notag
    \dot{A} 
    &= e^{-1}\nabla p^e + \Big(u - e^{-1}\textup{curl}\,B\Big)\times B 
    -  e^{-1}\mathcal{F} \\
    \notag 
    \mathcal{F} &= -\mathcal{J}_{\textup{ext}}\times B
\end{align}
with $B=\textup{curl}\,A$, and $p$ and $p^e$ are determined by $\textup{div}\,u = \textup{div}\,A = 0$. Without external control, $\mathcal{F} = 0$, these are the Hall MHD equations  (\cite{Lighthill60,Holm87,KMT21,O07}).

Let $\X_0(M)$ be the space of smooth vectorfields $X\in\X(M)$ such that $\textup{div}\,X=0$ and $\vv<X,n>|\del M=0$ where $n$ is the outward normal unit vector at the boundary $\del M$. Let $\mathbb{P}: \X(M)\to\X_0(M)$, $X\mapsto X-\nabla\Delta^{-1}\textup{div}\,X$ be the Leray-Hodge-Helmholtz projection ensuring the solenoidal property and the given boundary conditions. Thus \eqref{eS3:EulMaxCo} can be reformulated as 
\begin{align}
    \label{eS3:EulMaxCoP}
    \dot{u}
    &=  \mathbb{P}\Big( 
    -\vv<u,\nabla>u + \textup{curl}\,B\times B + \mathcal{F}\Big)\\
    \notag
    \dot{A} 
    &= 
    \mathbb{P}\Big(\Big(u - e^{-1}\textup{curl}\,B\Big)\times B 
    -  e^{-1}\mathcal{F}\Big)
\end{align}
with $B=\textup{curl}\,A$.

Given an unstable equilibrium $(u_e,A_e)$ of the free Hall MHD system  (i.e.\ $\mathcal{F}=0$) the goal is now to find a feedback law $\mathcal{F}=\mathcal{F}(u,A)$ such that $(u_e,A_e)$ is a stable equilibrium of \eqref{eS3:EulMaxCoP}. It is, however, not insisted that $\mathcal{F}$ should be of the form $-\mathcal{J}_{\textup{ext}}\times B$.

\subsection{Forced Lie-Poisson system}\label{subs:fLP}
Let $\D := \textup{Diff}_0(M) =: \G$ be the infinite dimensional Lie group of volume and boundary preserving diffeomorphisms. The groups $\D$ and $\G$ coincide as spaces but their roles are different, namely as in Section~\ref{sec1:feedback}, hence the distinction will be made throughout. 

The Lie algebras and (regular) dual algebras are $\du = \X_0(M) = \gu$ and $\du^* = \Om^1(M)/d\cinf(M) = \gu^*$ where $\Om^1(M)/d\cinf(M)$ is the regular dual consisting of one-forms modulo exact one-forms. The regular dual is isomorphic to the Lie algebra, and the isomorphism is given by the flat operator (lowering indices) followed by projecting onto the equivalence class modulo exact one-forms, $[\cdot]\circ\flat: \X_0(D)\to\Om^1(M)/d\cinf(M)$, $u\mapsto [u^{\flat}] = [\nu]$ with inverse $[\nu]\mapsto \mathbb{P}\nu^{\sharp} = u$;  
equivalence classes coincide, $[\nu] = [\kappa]$, if, and only if, $\nu-\kappa = df$ for $f\in\cinf(M)$.

The Lie bracket on $\du$ is the negative of the vector field bracket. That is, $\ad(u)v = -\vv<u,\nabla>v + \vv<v,\nabla>u$ for $u,v\in\du$. The coadjoint representation is given by the Lie derivative of one-forms, $\ad(u)^*[\nu] = [L_u\nu] = [di_u\nu + i_ud\nu]$ where $u\in\du$ and $[\nu]\in\du^*$.
Here $i_u$ is the insertion (contraction) operation and $d$ is the exterior derivative.
The adjoint and coadjoint represenations on $\gu$ and $\gu^*$ are given by the same formulas and the notation will also be identical. 

The equivalence class notation will be omitted from now on. E.g., $\nu\in\du^*$ will refer to $[\nu] \in\Om^1(M)/d\cinf(M)$ with the tacit understanding that $\nu$ is only fixed up to an exact one-form. 

Let 
\[
 \mu_0 := \flat: \du\to\du^*,\qquad
 \I_0 := -e^{2}\, \flat\circ\Delta^{-1},\qquad 
 \A_0 := -1
\]  
where $\Delta = -\textup{curl}^2: \X_0(M)\to\X_0(M)$ 
is the vector Laplacian (since $\textup{div}\,X = 0$ for $X\in\X_0(M)$).
The Kaluza-Klein type Hamiltonian function $h_0: \du^*\times\gu^*\to\R$ is defined as
\begin{align}
    h_0(\nu,\alpha)
    &= \frac{1}{2}\int_M||\nu+\alpha||^2\,dx 
     + \frac{1}{2e^2}\int_M||d\alpha||^2 \,dx 
     \\
\notag
    &=
    \frac{1}{2}\int_M
    \left\langle 
        \left(
        \begin{matrix}
         \mu_0^{-1} & -\mu_0^{-1}\A_0^* \\
         -\A_0\mu_0^{-1} & \I_0^{-1} + \A_0\mu_0^{-1}\A_0^*
        \end{matrix}
        \right) 
        \left(
        \begin{matrix}
         \nu \\
         \alpha 
        \end{matrix}
        \right) ,
        \left(
        \begin{matrix}
         \nu \\
         \alpha 
        \end{matrix}
        \right)
    \right\rangle\,dx
\end{align}
which has the same structure as \eqref{eS2:KKmetric}.
Let  $B=\textup{curl}\,A$ with $\textup{div}\,A = 0$. Then the identifications
\[
\nu = u^{\flat} + e A^{\flat}, \qquad
\alpha = - e A^{\flat}
\]
imply that equations \eqref{eS3:EulMaxCo} are equivalent to the system of forced Lie-Poisson equations
\begin{align}
    \label{eS3:fLP1}
    \dot{\nu} 
    &= -\ad(u)^*\nu, 
    \qquad u = \delta h_0/\delta\nu 
    = \mu_0^{-1}(\nu - \A_0^*\alpha) 
    \\
    \notag
    \dot{\alpha} 
    &= -\ad(X)^*\alpha + \mathcal{F}^{\flat} 
    ,\qquad  X = \delta h_0/\delta\alpha 
    = \I_0^{-1}\alpha - \A_0 u
\end{align}
on $\du^*\times\gu^*$.
The proof of this equivalency follows either by direct calculation or, from a structural point of view, by adapting the Hamiltonian structure for Hall MHD found by \cite{Holm87} to the incompressible case.

\subsection{Closed-loop dynamics and double bracket IDA-PBC}\label{subS3:loop} 
Concerning the choice of sign in Section~\ref{sec1:feedback}, the bottom sign is now the relevant one since \eqref{eS3:fLP1} is formulated in the space frame, that is $\pm1$ becomes $-1$.

In order to make use of Theorem~\ref{thm:FB_IDA}, we assume that  $\mathcal{F}^{\flat} = \mathcal{F}(u,A)^{\flat}  = \mathcal{U}_{\textup{LP}}(\nu,\alpha)+\widetilde{\mathcal{U}}_{\textup{diss}}(\nu,\alpha)$ as in definitions \eqref{eS1:tildeU} and \eqref{eS1:defUdiss} with $G=1$. The map $C: \du^*\to\gu^*$ and the sign $s=\pm1$ are to be defined in accordance with the control target. It follows that, in the new variables $(\nu,\beta) = \Phi(\nu,\alpha) = (\nu,\alpha+C\nu+N(\nu))$, system \eqref{eS3:fLP1} is equivalent to to the dissipative port controlled Hamiltonian system \eqref{eS1:pHS},   given by
\begin{align}
\label{eS3:lp1}
    \dot{\nu} 
    &= 
    -\ad(\mu_C^{-1}\nu)^*\nu 
    + \ad(\mu_0^{-1}\A_0^*\beta)^*\nu 
    - \ad(\mu_0^{-1}\A_0^*N(\nu))^*\nu  \\ 
\label{eS3:lp2}
    \dot{\beta} 
    &= 
    -N(\nu) - \beta 
\end{align}
with $N(\nu) = -s\I_0\A_0^*\mu_0^{-1}\ad(\mu_C^{-1}\nu)^*\nu$.  Definition~\eqref{eS1:def_g_C} yields 
\begin{equation}
    \label{eS3:g_C}
    g_C(\nu,\beta)
    = 
    \frac{1}{2}\vv<\nu,\mu_C^{-1}\nu > 
    +\frac{s}{2}\vv<\beta, \I_0^{-1}\beta> 
    = 
    \frac{1}{2}\int_M\vv<\nu,\mu_C^{-1}\nu > \,dx
    + \frac{s}{2e^2}\int_M ||d\beta||^2\,dx 
\end{equation}
where the angle bracket is used to denote, both, the $L^2$- and the pointwise Euclidean pairing. 

Depending on whether the goal is to  stabilize an energy minimum or maximum the sign is chosen as $s=1$ or $s=-1$, respectively.

\subsection{Controlled stability of planar shear flow}\label{Sec4:shear} 
Consider now the Hall MHD system~\eqref{eS3:EulMaxCoP} for the case of flow along the horizontal strip $M = \{(x,y): 0\le x\le L\pi , 0\le y \le W\pi  \}$, where length, $L\pi$, and width, $W\pi$, are fixed, in the presence of a vertical magnetic field, $B = b(x,y)\,e_3 = \textup{curl}\,A$. That is, 
\begin{align}
 \label{eS4:u}
 \dot{u}
 &= 
 -\nabla (p+\frac{1}{2}b^2)
 -\vv<u,\nabla u> + \mathcal{F} 
 = \mathbb{P}\Big( -\vv<u,\nabla u> + \mathcal{F} \Big) \\
  \notag 
 \dot{A}
 &=
  -e^{-1}\nabla (p^e-\frac{1}{2}b^2)
  + u\times B -  e^{-1}\mathcal{F} 
 = \mathbb{P}\Big( u\times B - e^{-1}\mathcal{F} \Big)
\end{align}
where $\mathbb{P}$ is the Leray-Hodge-Helmholtz projection ensuring  $\textup{div}\,u = \textup{div}\,A = 0$ and the boundary conditions. 

\begin{remark}
Planar incompressible Hall MHD flow coincides with ordinary MHD since $\textup{curl}\, B\times B = -\nabla b^2/2$ for a vertical field. But since \eqref{eS4:u} is derived as the two-dimensional version of \eqref{eS3:EulMaxCoP} it will also be referred to as controlled Hall MHD flow. 
\end{remark}

Let $0\le\gamma<1$ and consider equilibrium solutions, $u_e$ and $A_e$, of the uncontrolled equations, \eqref{eS4:u} with $E_{\textup{ext}}=0$, given by 
\begin{equation}
    \label{eS4:u_e}
    u_e = 
    \left( 
    \begin{matrix}
     \sin( y) \\
     0 
    \end{matrix}
    \right),  
    \qquad
    A_e = 
    -\gamma e^{-1}\left( 
    \begin{matrix}
     \sin( y) \\
     0 
    \end{matrix}
    \right) .
\end{equation} 
Since $\textup{curl}\,B\times B = -\nabla b^2/2$ is absorbed by the pressure term in the Euler equation the presence of an uncontrolled vertical magnetic field does not alter the stability properties of incompressible horizontal flow (\cite[Section~6.1]{HMRW85}).  Hence, for $\gamma=0$, this  equilibrium is stable if $\lam_1 > 1$ where $\lam_1 = 1/L^2 + 1/W^2$ is the minimal eigenvalue of $-\Delta$ in the domain $M$; it is unstable otherwise.  See \cite[Section~3.3]{HMRW85}. Given an arbitrarily long channel, the goal is to find $\gamma$ and $\F(u,A)$ such that $(u_e,A_e)$ is stable. 


Define $\Delta_{\gamma} := \del_x\del_x + (1-\gamma)\del_y\del_y$ and $\mu_C: \du\to\du^*$ by 
\begin{equation}
    \label{eS4:muC}
    \mu_C\left(\begin{matrix}
        v^1 \\ 
        v^2 
    \end{matrix}\right)
    = 
    \mu_0\Delta^{-1} \Delta_{\gamma}\left(\begin{matrix}
        v^1 \\ 
        v^2 
    \end{matrix}\right)
    = 
    \mu_0\mathbb{P}\left(\begin{matrix}
        (1-\gamma)v^1 \\ 
        v^2 
    \end{matrix}\right)
\end{equation}
where $\Delta$ and $\Delta_{\gamma}$ operate component wise on vector fields. Because of $\mu_C^{-1} = \mu_0^{-1}(1+\A_0^* C)$ and $\A_0=-1$ this implies $C = 1 - \mu_0\Delta_{\gamma}^{-1}\Delta\mu_0^{-1}$. 
Notice that 
\[
 C\Big( \nu^1(y)\,dx  \Big) 
 = -\frac{\gamma}{1-\gamma}\nu^1(y)\,dx . 
\]

\begin{remark}
The idea behind this choice of metric is that it makes the channel appear shorter as $\gamma$ approaches $1$ from below. Measured with respect to $\mu_C$ the channel length is $\sqrt{1-\gamma}L\pi$.    
\end{remark}

Consider now the corresponding feedback controlled system \eqref{eS4:u}  with control force 
\begin{align}
\label{eS4:mU}
    \mathcal{f}(u,A)
    &= 
    \mu_0^{-1}\Big(
     \mathcal{U}_{\textup{LP}}(\nu,\alpha)
     +\widetilde{\mathcal{U}}_{\textup{diss}}(\nu,\alpha)
    \Big)
    =
    \mu_0^{-1}\Big(
     \mathcal{U}_{\textup{LP}}(\nu,\alpha)
     +\mathcal{U}_{\textup{diss}}(\nu,\alpha+C\nu+N(\nu))
    \Big)\\
    \notag 
    &= 
    \mu_0^{-1}\Big(
      C\ad(u)^*\nu - dN(\nu)\,\ad(u)^*\nu 
     + 2 N(\nu) + \alpha + C\nu  
    \Big) 
\end{align}
The equations of motion in the $(\nu,\beta)=\Phi(\nu,\alpha)$ variables are \eqref{eS3:lp1}-\eqref{eS3:lp2}. The equilibrium is mapped to
\[
 (\nu_e,\beta_e) = \Phi(\nu_e, \alpha_e) = (\mu_c u_e, 0) 
\]
since $N(\nu_e)=0$ and $\alpha_e = -C\nu_e$. 

To show nonlinear stability of $(\nu_e,\beta_e)$ we use the energy-Casimir method. 
The first variation of $g_C$ at the equilibrium is 
\begin{equation}
    D^1g_C(\nu_e,\beta_e)(\delta \nu,\delta\beta)
    = \vv<\delta\nu, \mu_C^{-1}\nu_e> 
    = \vv<\delta\nu, u_e> 
    = -\int_M \psi_e(\star d \delta\nu)\,dx dy
    = -\int_M \psi_e \delta\om\,dx dy
\end{equation}
where $\psi_e$ is the stream function of $u_e$, $u_e = \nabla^s\psi_e = (-\del_y\psi_e,\del_x\psi_e)$, and $\star$ is the Hodge star operator, $\star f(x,y)dx\wedge dy = f(x,y)$, and $\star d \delta\nu = \delta\om$ is the vorticity associated to $\delta\nu$.
This expression is nonzero whence a Casimir function, $K_C$, is needed such that $D^1(g_C+K_C)(\nu_e,\beta_e) = 0$. A suitable choice is 
\begin{equation}
    K_C(\nu,\beta) 
    = K_C(\nu) 
    = -\frac{1}{2(1-\gamma)}\int_M(\star d \nu)^2\,dxdy
\end{equation}
where $\star d \nu$ is the vorticity associated to $\nu$. Vorticity integrals are Casimir functions on $\du^*$, hence $K_C$ is constant along solutions of \eqref{eS3:lp1}. 
 Coadjoint orbits in $\du^*$ are of the form $\mathcal{O}_{\nu_e} = \{\nu\in\du^*: *d\nu = (*d\nu_e)\circ\phi,\; \phi\in\D\}$, i.e., those elements whose vorticities are related by an area preserving diffeomorphism. 
 
Moreover, $\om_e = \star d \nu_e = \star d \mu_C\nabla^s \psi_e = \Delta_{\gamma}\psi_e$ and $\psi_e$ are related by $\om_e = -(1-\gamma)\psi_e$. Therefore,
\begin{equation}
    D^1(g_C+K_C)(\nu_e,\beta_e)(\delta \nu,\delta\beta)
    = -\int_M \psi_e \delta\om\,dx dy
    - \frac{1}{1-\gamma}\int_M \om_e \delta\om\,dx dy
    = 0.
\end{equation}
Let $\delta\psi$ be the stream function associated to a perturbation $\delta\nu$, i.e.\ $\delta\nu = \mu_C\nabla^s\delta\psi$, and $\delta\om = \star d \delta\nu = \Delta_{\gamma}\delta\psi$ the vorticity. 
The second variation is now estimated by the Poincar\'e inequality as
\begin{align*}
    D^2(g_C+K_C)(\nu_e,\beta_e)((\delta\nu,\delta\beta),(\delta\nu,\delta\beta))
    &=
    \vv<\delta\nu, \mu_C^{-1}\delta\nu> 
    - \frac{1}{1-\gamma}\int_M (\delta\om)^2\,dx dy
    +s \vv<\delta\beta,\I_0^{-1}\delta\beta>  \\ 
    &= 
    - \int_M \delta\psi\Delta_{\gamma}\delta\psi \,dxdy
    - \frac{1}{1-\gamma}\int_M (\delta\om)^2\,dx dy
    +s \vv<\delta\beta,\I_0^{-1}\delta\beta>  \\ 
    &\le 
    \lam_1(\gamma)^{-1}\int_M (\Delta_{\gamma}\delta\psi)^2 \,dxdy
    - \frac{1}{1-\gamma}\int_M (\delta\om)^2\,dx dy
    +s \vv<\delta\beta,\I_0^{-1}\delta\beta>
\end{align*}
where 
\[ 
 \lam_1(\gamma) = \frac{1}{L^2} + \frac{1-\gamma}{W^2}
\]
is the smallest eigenvalue of $-\Delta_{\gamma}$ on $M$. 
The quadratic form $D^2(g_C+K_C)(\nu_e,\beta_e)$ may be indefinite for $\gamma=0$ due to the Casimir contribution. However, it is negative definite if 
\begin{align}
\label{eS3:stC}
 \frac{1}{(1-\gamma)L^2} + \frac{1}{W^2}  &> 1     \\
 \notag 
 s &= -1
\end{align}
With these choices, $L_C(\nu,\beta) := -g_C(\nu,\beta) - K_C(\nu) + g_C(\nu_e,\beta_e) + K_C(\nu_e)$ is a  Lyapunov function: it is positive in a punctured neighborhood of $(\nu_e,\beta_e)$ and decreasing along solutions. 
Let $(\nu_t, \beta_t) = (\nu_e+\delta\nu_t, \beta_e+\delta\beta_t)$ be a perturbed solution of \eqref{eS3:lp1}-\eqref{eS3:lp2}. It follows that 
\[
 L_C(\nu_e+\delta\nu_0,\beta_e+\delta\beta_0) \ge L_C(\nu_t, \beta_t)
 \ge \lam_1(\gamma)^{-1}\Big(\frac{1}{(1-\gamma)L^2}+\frac{1}{W^2}-1\Big)\int_M(\delta\om_t)^2
 + \int_M(\nabla\delta\beta_t)^2
\]
and hence $(\nu_e,\beta_e)$ is a nonlinearly stable equilibrium, with respect to the Sobolev $H^1$ norm, when conditions~\eqref{eS3:stC} are satisfied.

Moreover, \eqref{eS1:g_C} implies that 
\[
 \frac{\del}{\del t}L_C(\nu_t,\beta_t)
 = -\frac{\del}{\del t}g_C(\nu_t, \beta_t)
 = -||N(\nu_t)||_{\mathbb{I}_0^{-1}}^2 -||\delta\beta_t||_{\mathbb{I}_0^{-1}}^2
\]
converges to $0$ since $L_C$ is nonnegative and decreasing along solutions. Hence $N(\nu_t)\to0$ and  $\beta_t\to\beta_e=0$ asymptotically in the $H^1$-norm. It is not obvious that this implies also \revise{convergence of $\nu_t$ to an equilibrium $\widetilde{\nu}_e$,} 
hence asymptotic stability is not concluded. 

However, \revise{even if $\nu_t\to\widetilde{\nu}_e$} it need not be the case that $\delta\om_t\to0$, \revise{i.e.\ $\widetilde{\nu}_e = \nu_e$}: the reason for this is the same as in the satellite example in Section~\ref{sec2:sat}, perturbations to different coadjoint orbits in $\du^*$ cannot be restored since this structure is preserved by the dynamics.

%
 

\section{Comparison to CL and IDA-PBC techniques}
This paper rests on ideas from CL theory and IDA-PBC techniques, it therefore makes sense to compare the results to these two approaches. 

Concerning CL theory, and specifically the satellite example, we note that  \cite{BMS97} conclude stability of the controlled satellite by showing that it is equivalent to a Hamiltonian system, up to a certain factor (cf.\ sentence immediately below \cite[Equ.~(3.7)]{BMS97}). In the present approach, there is no need to introduce such a factor ex post due to the isomorphism $G$. Thus in contrast to `classical' CL theory, the matching construction of Proposition~\ref{prop:matching} yields a direct identification of the feedback controlled and the corresponding Lie-Poisson systems. However, for the conclusion of stability the manner of this identification does, at least for the satellite example, not make any difference. 

We emphasize that stability of the controlled satellite was shown in  Section~\ref{sec:sat_AS} with respect to the feedback law \eqref{eS1:LPdissXPL} which \emph{does not} satisfy the matching conditions. Thus this control approach is independent of matching conditions which need not always have a (non-trivial) solution. 

The same holds for the conclusion of asymptotic stability. In \cite{BCLMW00} this was shown by first proving stability via the matching construction and then adding a dissipative control to achieve asymptotic stability. Again, the present approach is quite different since it is not only independent of matching conditions but also because dissipation is included by the double bracket construction (via the nonlinear map $N$ in equation~\eqref{eS1:LPdiss1}) and a restoring force in the closed loop variable $\beta$ (via $\mathcal{U}_{\textup{diss}}$ in equation~\eqref{eS1:LPdiss2}). 

Concerning IDA-PBC approaches to Hamiltonian systems, it is not a priori clear how to include symmetry arguments in the construction of the interconnection and damping structures. In the present context this rests on the geometric formulation of the control force in Proposition~\ref{prop:Cons} which encodes the symmetries ab initio in the desired manner. From this we arrive at the closed loop equations \eqref{eS1:pH1}-\eqref{eS1:pH2} constructively such that the IDA-PBC structure may be read off since these equations are linear in the first order derivatives of the quadratic function $g_C$ defined in \eqref{eS1:def_g_C}. 

Thus, while the construction of Theorem~\ref{thm:FB_IDA} is covered by the IDA-PBC methodology from an abstract point of view, the benefit of the present geometric approach is that it is constructive and makes use of available symmetries. 

Moreover, the coordinate independence of the geometric formulation allows to apply Theorem~\ref{thm:FB_IDA} to the infinite dimensional problem of feedback stabilization of Hall MHD flow (Section~\ref{Sec3}). This system has been  treated previously neither from the CL nor from the IDA-PBC point of view, although there are many works concerned with feedback control and (linear) stability of (Hall) MHD flow (e.g., \cite{HT04,GK06,XSVK08,VSK09,Tassi22}). The results of this paper are restricted to Lie-Poisson systems defined on direct product Lie algebras. A superficial inspection shows that the derivation of Theorem~\ref{thm:FB_IDA} depends only on structural arguments which may also be formulated for $\G$-invariant Hamiltonian systems on the cotangent bundle of a general principal fiber bundle $\P\to\D$ with structure group $\G$, when the Hamiltonian is associated to a $\G$-invariant Riemannian metric on $\P$ and there exists natural principal bundle connection specifying the control directions. Thus it may be hoped that the approach can be generalized such that other examples can be treated in a similar manner. 

\vspace{1cm}
\textbf{Acknowledgements.} 
The reviewer reports are gratefully acknowledged. 

\vspace{1cm}
\textbf{Declaration.}
The author has no competing interests to declare that are relevant to the content of this article.


\begin{thebibliography}{99}


\bibitem{AML16}
Eric C. D'Avignon,  Philip J. Morrison, and Manasvi Lingam, 
\emph{Derivation of the Hall and extended magnetohydrodynamics brackets},
Physics of Plasmas \textbf{23}, 062101 (2016); \href{https://doi.org/10.1063/1.4952641}{https://doi.org/10.1063/1.4952641}

\bibitem{BOS02}
Blankenstein, Guido; Ortega, Romeo; Schaft, Arjan J. van der;
\emph{The matching conditions of controlled Lagrangians and IDA-passivity based control}, 
International Journal of Control \textbf{75} (2002) - Issue 9.
\href{https://doi.org/10.1080/00207170210135939}{https://doi.org/10.1080/00207170210135939}

\bibitem{BCLMW00}
Anthony M. Bloch, Dong Eui Chang, Naomi E. Leonard, Jerrold E. Marsden, Craig Woolsey,
\emph{Asymptotic Stabilization of Euler-Poincar\'e Mechanical Systems},
IFAC Proceedings Volumes,
Volume \textbf{33}, Issue 2 (2000).
\href{https://doi.org/10.1016/S1474-6670(17)35546-5}{https://doi.org/10.1016/S1474-6670(17)35546-5}

\bibitem{BLM01a}
A. Bloch, N. Leonard, J. Marsden,
\emph{Controlled Lagrangians and the Stabilization of
Mechanical Systems I: The First Matching Theorem},
IEEE Trans. on Sytems and Control, \textbf{45}, (2001), 2253-2270.
\href{https://doi.org/10.1109/9.895562}{https://doi.org/10.1109/9.895562}

\bibitem{BLM01b}
A. Bloch, N. Leonard, J. Marsden,
\emph{Controlled Lagrangians and the stabilization of
Euler-Poincar\'e mechanical systems},
Int. J. Robust Nonlinear Control (2001) \textbf{11}:191-214.  
\href{https://doi.org/10.1002/rnc.572}{https://doi.org/10.1002/rnc.572}


\bibitem{BKMR96}
Anthony Bloch, P. S. Krishnaprasad, Jerrold E. Marsden, Tudor S. Ratiu,
\emph{The Euler-Poincar\'e equations and double bracket dissipation}
Comm. Math. Phys. 175(1): 1-42 (1996).
\href{https://doi.org/10.1007/BF02101622}{https://doi.org/10.1007/BF02101622}

\bibitem{BKMS92}
A.M. Bloch, P.S. Krishnaprasad, J.E. Marsden, G. S\'anchez de Alvarez,
\emph{Stabilization of rigid body dynamics by internal and external torques},
Automatica \textbf{28}, Issue 4 (1992), Pages 745-756.
\href{https://doi.org/10.1016/0005-1098(92)90034-D}{https://doi.org/10.1016/0005-1098(92)90034-D}

\bibitem{BMS97}
Bloch A.M., Marsden J.E., S\'anchez de Alvarez G. (1997), 
\emph{Feedback Stabilization of Relative Equilibria for Mechanical Systems with Symmetry.} In: Alber M., Hu B., Rosenthal J. (eds) \emph{Current and Future Directions in Applied Mathematics}. Birkh\"auser, Boston, MA.
\href{https://doi.org/10.1007/978-1-4612-2012-1\_11}{https://doi.org/10.1007/978-1-4612-2012-1\_11}

\bibitem{BCO16}
Pablo Borja, Rafael Cisneros, Romeo Ortega,
\emph{A constructive procedure for energy shaping of port—Hamiltonian systems},
Automatica \textbf{72} (2016), Pages 230-234.
\href{https://doi.org/10.1016/j.automatica.2016.05.028}{https://doi.org/10.1016/j.automatica.2016.05.028}

\bibitem{CBLMW02}
Chang, Dong Eui; Bloch, Anthony M.; Leonard, Naomi E.; Marsden, Jerrold E.; Woolsey, Craig A.;
\emph{The equivalence of controlled lagrangian and controlled hamiltonian systems},
ESAIM: Control, Optimisation and Calculus of Variations, Tome \textbf{8} (2002) , pp. 393-422.
\href{https://doi.org/10.1051/cocv:2002045}{https://doi.org/10.1051/cocv:2002045}

\bibitem{CM04}
Chang, Dong Eui; Marsden, Jerrold E.;
\emph{Reduction of Controlled Lagrangian and Hamiltonian Systems with Symmetry}
SIAM J. Control Optim., 43(1), 277–300 (2004). 
\href{https://doi.org/10.1137/S0363012902412951}{https://doi.org/10.1137/S0363012902412951}

\bibitem{FS01}
Kenji Fujimoto, Toshiharu Sugie,
\emph{Canonical transformation and stabilization of generalized Hamiltonian systems},
Systems \& Control Letters \textbf{42}, Issue 3
(2001).
\href{https://doi.org/10.1016/S0167-6911(00)00091-8}{https://doi.org/10.1016/S0167-6911(00)00091-8}

\revise{
\bibitem{GK06}
Roland Griesse and Karl Kunisch, \emph{Optimal Control for a Stationary MHD System in Velocity‐Current Formulation}
SIAM J. Control Optim., \textbf{45}(5), 1822–1845 (2006). 
\href{https://doi.org/10.1137/050624236}{https://doi.org/10.1137/050624236}
}

\revise{
\bibitem{HT04}
E. Hameiri, R. Torasso,
\emph{Linear stability of static equilibrium states in the Hall-magnetohydrodynamics model}, 
Physics of Plasmas \textbf{11}, 4934–4945 (2004)
\href{https://doi.org/10.1063/1.1784453}{https://doi.org/10.1063/1.1784453}
}

\bibitem{H19}
S. Hochgerner, 
\emph{Feedback control of charged ideal fluids}, Nonlinearity \textbf{34} (2021) Number 3.
\href{https://doi.org/10.1088/1361-6544/abbd83}{https://doi.org/10.1088/1361-6544/abbd83}

\bibitem{H20}
S. Hochgerner, \emph{Symmetry actuated closed-loop Hamiltonian systems}, 
J. Geometric Mechanics (2020) \textbf{12}(4): 641-669. \href{http://dx.doi.org/10.3934/jgm.2020030}{http://dx.doi.org/10.3934/jgm.2020030}

\bibitem{HMRW85}
D. Holm, J. Marsden, T. Ratiu, A. Weinstein,
\emph{Nonlinear stability of fluid and plasma equilibria}, 
Physics reports \textbf{123} (1-2), 1-116 (1985).
\href{https://doi.org/10.1016/0370-1573(85)90028-6}{https://doi.org/10.1016/0370-1573(85)90028-6}

\bibitem{Holm87}
D. Holm, 
\emph{Hall magnetohydrodynamics: Conservation laws and Lyapunov stability},
The Physics of Fluids \textbf{30}, 1310 (1987); \href{https://doi.org/10.1063/1.866246}{https://doi.org/10.1063/1.866246}

\bibitem{HSS09}
Darryl D. Holm, Tanya Schmah, and Cristina Stoica, 
\emph{Geometric Mechanics and Symmetry.
From Finite to Infinite Dimensions}, Oxford University Press 2009. 

\bibitem{KMT21}
Kaltsas, D., Throumoulopoulos, G.,  Morrison, P. 
\emph{Hamiltonian kinetic-Hall magnetohydrodynamics with fluid and kinetic ions in the current and pressure coupling schemes}. Journal of Plasma Physics \textbf{87}(5).
\href{https://doi.org/10.1017/S0022377821000994}{https://doi.org/10.1017/S0022377821000994}

\bibitem{K85}
P.S. Krishnaprasad,
\emph{Lie-Poisson structures, dual-spin spacecraft and asymptotic stability},
Nonlinear Analysis: Theory, Methods \& Applications
Volume \textbf{9}, Issue 10, (1985), pp. 1011-1035. 
\href{https://doi.org/10.1016/0362-546X(85)90083-5}{https://doi.org/10.1016/0362-546X(85)90083-5}

\bibitem{Lighthill60}
M.J. Lighthill, 1960
\emph{Studies on Magneto-Hydrodynamic Waves and other Anisotropic wave motionsPhilosophical}, 
Trans. Royal Soc. London. Series A, Mathematical and Physical Sciences 252:397–430 (1960).
\href{http://doi.org/10.1098/rsta.1960.0010}{http://doi.org/10.1098/rsta.1960.0010}


\bibitem{MRW84}
Marsden, Jerrold E., Tudor Ratiu, and Alan Weinstein. 
\emph{Semidirect Products and Reduction in Mechanics.} Trans. Amer. Math. Soc. \textbf{281}, no. 1 (1984): 147–77. \href{https://doi.org/10.2307/1999527}{https://doi.org/10.2307/1999527}.

\bibitem{MSKS17}
Rachit Mehra, Sumeet G. Satpute, Faruk Kazi, N.M. Singh,
\emph{Control of a class of underactuated mechanical systems obviating matching conditions},
Automatica  \textbf{86} (2017), Pages 98-103.
\href{https://doi.org/10.1016/j.automatica.2017.07.033}{https://doi.org/10.1016/j.automatica.2017.07.033}

\bibitem{Michor06}
P. Michor, \emph{Some Geometric Evolution Equations  Arising as Geodesic Equations on Groups of Diffeomorphism, Including the Hamiltonian Approach.} IN: Phase space analysis of Partial Differential Equations. Series: Progress in Non Linear Differential Equations and Their Applications, Vol. 69. Bove, Antonio; Colombini, Ferruccio; Santo, Daniele Del (Eds.). Birkh\"auser Verlag 2006. Pages 133-215. 

\bibitem{O07}
Ohsaki, S., \emph{Variational Principle of Hall Magnetohydrodynamics}. J Fusion Energ \textbf{26}, 135–137 (2007). \href{https://doi.org/10.1007/s10894-006-9029-2}{https://doi.org/10.1007/s10894-006-9029-2}

\bibitem{OSME02}
Romeo Ortega, Arjan van der Schaft, Bernhard Maschke, Gerardo Escobar,
\emph{Interconnection and damping assignment passivity-based control of port-controlled Hamiltonian systems},
Automatica \textbf{38} (2002) 585 – 596. 
\href{https://doi.org/10.1016/S0005-1098(01)00278-3}{https://doi.org/10.1016/S0005-1098(01)00278-3}

\bibitem{OG04}
Romeo Ortega, Eloísa García-Canseco, \emph{Interconnection and Damping Assignment Passivity-Based Control: A Survey},
European Journal of Control \textbf{10} Issue 5
(2004), 432-450. 
\href{https://doi.org/10.3166/ejc.10.432-450}{https://doi.org/10.3166/ejc.10.432-450}

\revise{
\bibitem{Tassi22}
Emanuele Tassi, \emph{Formal stability in Hamiltonian fluid models for plasmas}, Journal of Physics A:
Mathematical and Theoretical  \textbf{55} (41), pp.413001 (2022). 
\href{https://doi.org//10.1088/1751-8121/ac8f76}{https://doi.org//10.1088/1751-8121/ac8f76}
}



\revise{
\bibitem{VSK09}
R. Vazquez, E. Schuster, and M. Krstic, \emph{A closed-form full-state feedback controller for stabilization of 3D magnetohydrodynamic channel flow}, Journal of Dynamic Systems, Measurement, and Control, vol. \textbf{131}, (2009).
\href{https://doi.org/10.1115/1.3089561}{https://doi.org/10.1115/1.3089561}
}

\bibitem{WK92}
L. Wang, P.S. Krishnaprasad, 
\emph{Gyroscopic control and stabilization},
J Nonlinear Sci \textbf{2}, 367–415 (1992). \href{https://doi.org/10.1007/BF01209527}{https://doi.org/10.1007/BF01209527}

\revise{
\bibitem{XSVK08}
Chao Xu, Eugenio Schuster, Rafael Vazquez, Miroslav Krstic,
\emph{Stabilization of linearized 2D magnetohydrodynamic channel flow by backstepping boundary control},
Systems \& Control Letters,
Volume \textbf{57}, Issue 10 (2008),
Pages 805-812.
\href{https://doi.org/10.1016/j.sysconle.2008.03.008}{https://doi.org/10.1016/j.sysconle.2008.03.008}
}

\end{thebibliography}
\end{document}